\documentclass[preprint,journal]{vgtc}

\newtheorem{theorem}{Theorem}[section]
\newenvironment{proof}{\paragraph{\textit{Proof}.}}{\hfill$\square$}

\ifpdf%
  \pdfoutput=1\relax                   %
  \usepackage{graphicx}                %
  \DeclareGraphicsExtensions{.pdf,.png,.jpg,.jpeg} %
\else%
  \usepackage{graphicx}                %
  \DeclareGraphicsExtensions{.eps}     %
\fi%

\graphicspath{{figures/}{pictures/}{images/}{./}} %

\vgtcinsertpkg

\usepackage{microtype}                 %
\PassOptionsToPackage{warn}{textcomp}  %
\usepackage{textcomp}                  %
\usepackage{mathptmx}                  %
\usepackage{times}                     %
\usepackage{cite}                      %
\usepackage{tabu}                      %
\usepackage{booktabs}                  %
\usepackage{paralist} 
\usepackage{amsmath}
\usepackage{amssymb}
\usepackage[T1]{fontenc}
\usepackage{xurl}
\usepackage{svg}
\usepackage{subcaption}
\usepackage{multirow}
\usepackage{todonotes}

\newcommand{\change}[1]{\textcolor{black}{#1}}

\onlineid{0}

\vgtccategory{Research}
\vgtcpapertype{algorithm/technique}

\title{MosaicSets: Embedding Set Systems into Grid Graphs}

\author{Peter Rottmann, Markus Wallinger, Annika Bonerath, Sven Gedicke, Martin Nöllenburg, and Jan-Henrik Haunert}

\authorfooter{
\item
 Peter Rottmann, Annika Bonerath, Sven Gedicke, and Jan-Henrik Haunert are with the Geoinformation Group of the University of Bonn. E-mail: \{surname\}@igg.uni-bonn.de.
\item
 Markus Wallinger and Martin Nöllenburg are with the Algorithms and Complexity Group of the Technical University of Vienna. E-mail: \{mwallinger,noellenburg\}@ac.tuwien.ac.at.
}

\shortauthortitle{Rottmann \MakeLowercase{\textit{et al.}}: MosaicSets: Embedding Set Systems into Grid Graphs}
\abstract{
Visualizing sets of elements and their relations is an important research area in information visualization. In this paper, we present \emph{MosaicSets}: a novel approach to create Euler-like  diagrams from non-spatial set systems such that each element occupies one cell of a regular hexagonal or square grid. The main challenge is to find an assignment of the elements to the grid cells such that each set constitutes a contiguous region. As use case, we consider the research groups of a university faculty as elements, and the departments and joint research projects as sets. We aim at finding a suitable mapping between the research groups and the grid cells such that the department structure forms a base map layout. Our objectives are to optimize both the compactness of the entirety of all cells and of each set by itself. We show that computing the mapping is NP-hard. However, using integer linear programming we can solve real-world instances optimally within a few seconds. Moreover, we propose a relaxation of the contiguity requirement to visualize otherwise non-embeddable set systems. 
We present and discuss different rendering styles for the set overlays. Based on a case study with real-world data, our evaluation comprises quantitative measures as well as expert interviews. 
}


\CCScatlist{ %
 \CCScat{K.6.1}{Management of Computing and Information Systems}%
{Project and People Management}{Life Cycle};
 \CCScat{K.7.m}{The Computing Profession}{Miscellaneous}{Ethics}
}

\teaser{
  \centering
  $\vcenter{\hbox{\includegraphics[scale=0.27]{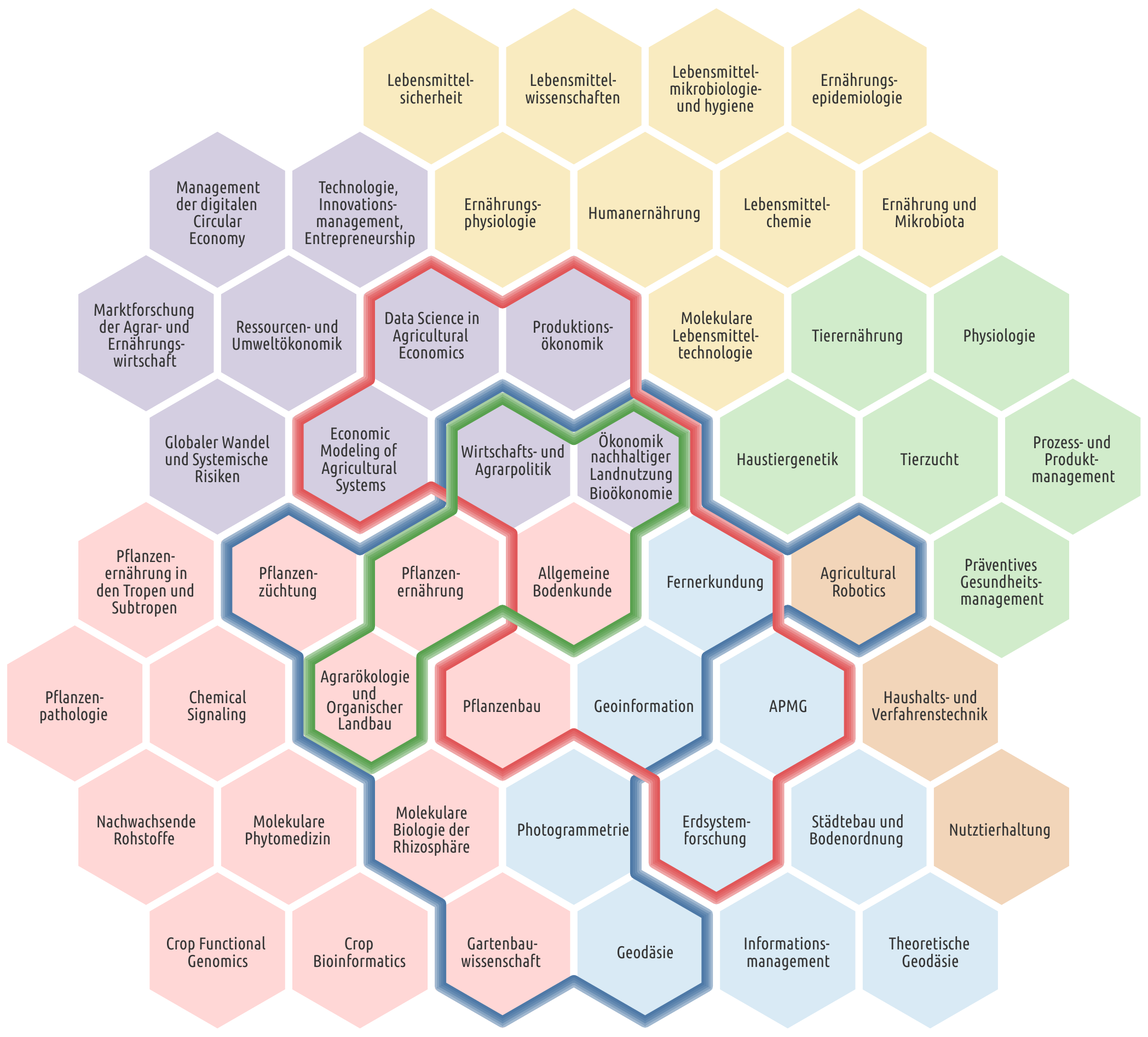}}}$\quad
  $\vcenter{\hbox{\includegraphics[scale=0.27]{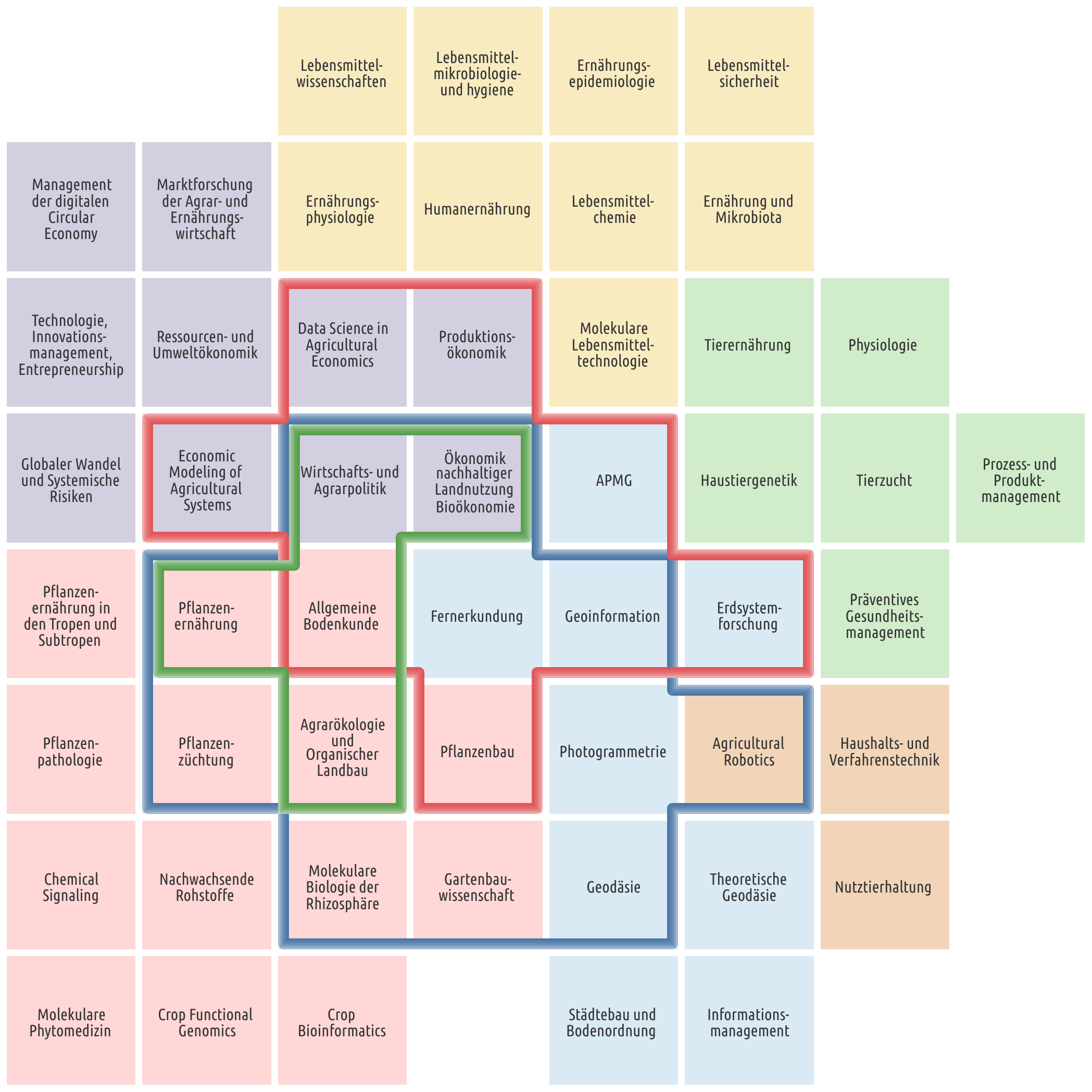}}}$\quad
  $\vcenter{\hbox{\includegraphics[scale=0.08]{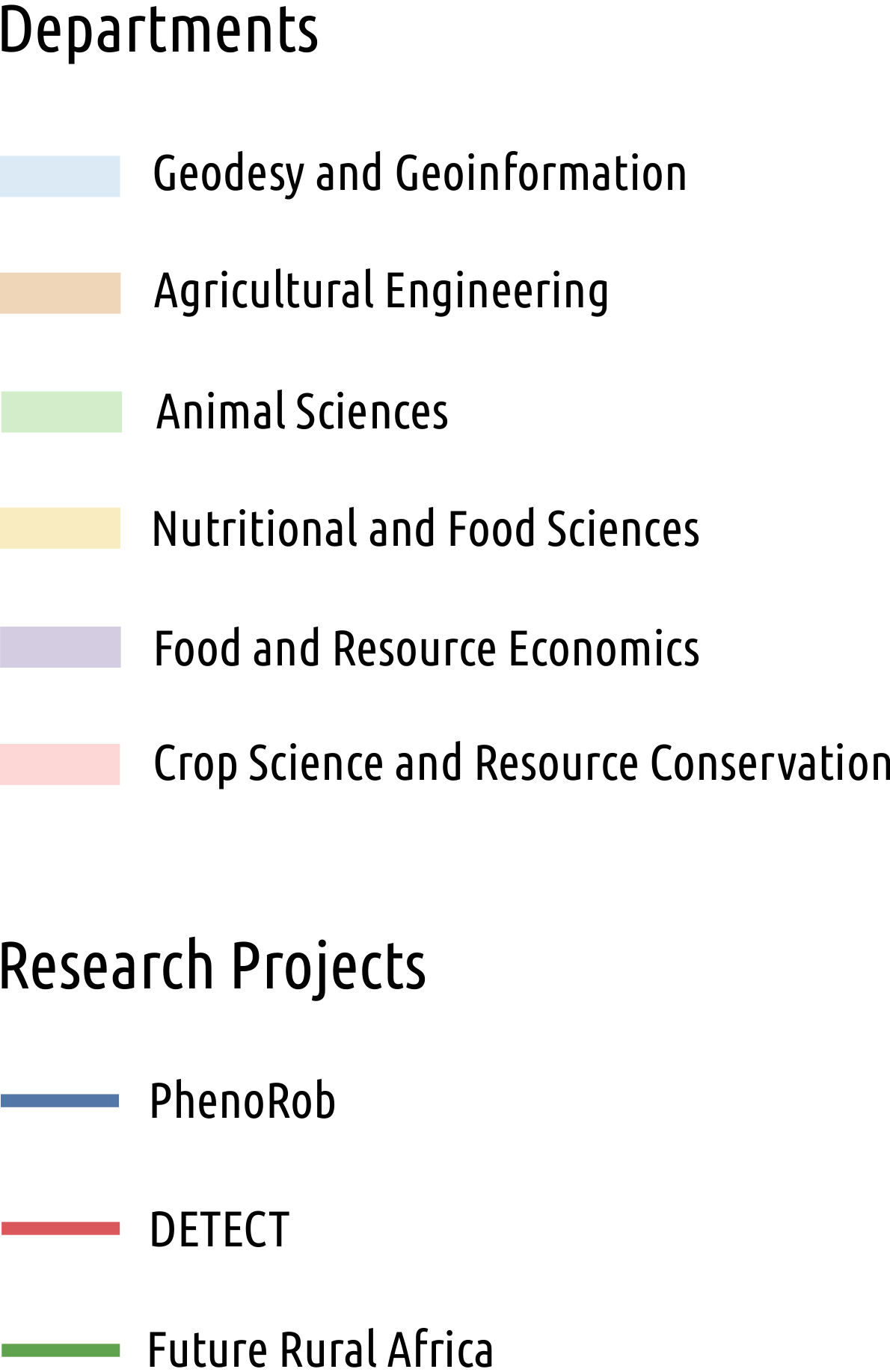}}}$
  \caption{Visualizing the research groups of the Agricultural Faculty of the University of Bonn with \emph{MosaicSets} using the model variant MSE optimizing eccentricity-based compactness of sets (see Sect.~\ref{ssec:setup}).}
  \label{fig:teaser} 
}

\renewcommand{\manuscriptnotetxt}{}

\nocopyrightspace

\newsavebox\mybox

\begin{document}

\firstsection{Introduction}
\label{sec:intro}

\maketitle
Set visualization is an important branch of information visualization that generally deals with the visualization of set systems, i.e., the relationships between multiple sets (e.g., set intersections, inclusions, exclusions) and of individual elements with their set memberships.
As a concrete use case and our running example, we consider the scenario that the board of an institution (e.g., a university faculty) aims to improve its strategic planning processes by supporting the discussions at meetings with informative visualizations. 
In particular, the institution board wishes to visualize the institution's partition into organizational units (e.g., departments and research groups) as well as important intra-institutional collaborations (e.g., larger research projects involving multiple research groups of different departments). 
In this scenario, the research groups are the elements of the set system and the departments as well as the intra-institutional collaborations are the sets. 
In particular, each research group belongs to exactly one department and possibly one or more intra-institutional collaborations. 
In consequence, the departments form a disjoint set cover of all research groups. 
In MosaicSets, we will exploit this property of having a partition of all elements into a disjoint set cover, which is a property that many real-world set systems have.

For visualizing abstract set systems, Euler and Venn diagrams are commonly used and several methods for computing them have been proposed.
However, they usually focus on the set relationships and, even in area-proportional diagrams, either do not show individual elements at all or they do not play a primary role in the diagrams.
In our use case, though, elements such as the research groups are of primary importance for the envisioned tasks and thus it is desirable that each element is prominently represented with an area of the same size and shape, guaranteeing an equitable visual representation of each research group.
Since the existing Euler-like methods do not satisfy this requirement, we present a new approach that embeds a given set system into a prescribed grid, which for example may consist of square-shaped or hexagonal tiles, such that each element occupies exactly one grid tile (or grid cell).
Our primary goal is to ensure that for each of the given sets the corresponding grid cells constitute a contiguous region as in an Euler diagram.
Coloring the cells according to the partition into a disjoint set cover, e.g., by organizational unit of the assigned element, we obtain a visualization similar to a grid map \cite{DBLP:journals/ijcga/EppsteinKSS15} with uniform cells and schematic region boundaries; see Fig.~\ref{fig:teaser}.
This serves as the base map of our set system. 
Sets representing intra-institutional collaborations can then be visualized as contiguous overlays on top of the base map, for example, using colored contours or visual links as Kelp diagrams~\cite{dinkla2012kelp,meulemans2013kelpfusion}.
If the number of such collaborations is small, they all may be visualized at the same time. 
However, even if at any time only one or few selected collaborations are shown, we consider it beneficial to compute a single base map embedding once under consideration of all collaborations, which can then be shown on demand in an interactive setting or using a small-multiples approach.
That way the grid map of organizational units provides a stable context for visualizations of different collaborations.

Grid maps %
have been introduced primarily for proper geovisualization tasks. 
In contrast, we here use the map as a metaphor to visualize non-geographical data, namely an abstract set system, in an intuitive way.
While methods for the computation of map-like visualizations of graphs \cite{DBLP:conf/gd/GansnerHK09} and set systems \cite{ehkp-mvecg-15} have been presented before, using grid maps for set visualization is new.
Consequently, with this paper we introduce \emph{MosaicSets} as a new set visualization technique.
We aim to demonstrate the potential of MosaicSets for the described use case but also to open up a new topic of research for the network and set visualization community.
As concrete contributions, we present 
\begin{compactitem}
\item our design decisions and a corresponding formal mathematical model for optimizing the MosaicSets visualization as a constrained hypergraph embedding problem, alongside with a proof of its NP-completeness (Sect.~\ref{sec:formalization}),
\item an exact method for solving the hypergraph embedding problem using integer linear programming including variations to optimize  compactness or relax contiguity requirements if needed (Sect.~\ref{sec:ilp}), 
\item an implementation providing a choice of square or hexagonal grids for the base map and two overlay rendering styles for highlighting the additional sets (Sect.~\ref{sec:rendering}), and
\item a qualitative evaluation based on an expert interview for the real use case, which has inspired this work: the strategy development of the Agricultural Faculty of the University of Bonn; this is complemented by quantitative computational experiments, which explore the performance and resulting quality of MosaicSets (Sect.~\ref{sec:evaluation}). 
\end{compactitem}
We invite the reader to explore our two interactive visualizations under \url{https://www.geoinfo.uni-bonn.de/mosaicsets} \change{and our code at \url{https://gitlab.igg.uni-bonn.de/geoinfo/mosaicsets}}.

\section{Related Work}
Several topics in the literature relate to concepts relevant for MosaicSets; we discuss general set visualization (Sect.~\ref{sec:relwork-setvis}), geovisualization (Sect.~\ref{sec:relwork-geovis}), map metaphors for abstract data visualization (Sect.~\ref{sec:relwork-graphRepresentation}), districting problems (Sect.~\ref{sec:relwork-districting}), and hypergraph drawings via support graphs (Sect.~\ref{sec:relwork-supportgraphs}).

\subsection{Set Visualization}\label{sec:relwork-setvis}
We follow the classification of set visualization techniques by Alsallakh et al.~\cite{DBLP:journals/cgf/AlsallakhMAHMR16} and discuss MosaicSets with respect to a selection of these.

Two standard techniques for set visualization are Euler and Venn diagrams. Both represent each set by a closed curve and its enclosed region. In Euler diagrams, every overlap of a set of regions represents a non-empty intersection of the corresponding sets. Labels for elements can be placed in the overlap areas, however, this is not necessary to represent the set relations.
In contrast, Venn diagrams show an overlap area for every possible combination of set intersections; here labels for elements are necessary to indicate which set intersections are non-empty.
Euler diagrams are an intuitive way to displaying elements, sets, and set relations, but are mostly limited to a few sets due to clutter and drawability issues~\cite{DBLP:journals/cgf/AlsallakhMAHMR16}. 
Many Euler diagram techniques~\cite{StapletonRHZ11,RodgersZF08,MicallefR14,KehlbeckGWD22} focus on showing set relations and hence do not display individual elements. 
There are also examples of Euler diagram techniques that explicitly visualize also the elements as small glyphs or by distributing the element names inside the \change{corresponding} regions~\cite{micallef2012assessing,SimonettoAA09,brath2012multi,riche2010untangling}.
MosaicSets is a special form of an Euler diagram. It differs from existing methods by visualizing every element as a uniformly-shaped and equally-sized grid cell. In contrast to other approaches, our use case leads to a more refined problem setting. In particular, while our method for the embedding step can deal with arbitrary set systems, our method for the subsequent rendering step expects that a part of the sets (e.g., the departments) forms a partition of the elements. Their visualization provides a background for the other sets (e.g., the projects).

Matrix-based techniques are set visualizations that have the advantage of being clutter-free~\cite{DBLP:journals/cgf/AlsallakhMAHMR16}. An example is ConSet~\cite{kim2007visualizing}, where the sets and elements are associated to the rows and columns of a matrix; matrix entries encode whether an element is contained in a set or not.
In UpSet~\cite{lex2014upset}, the matrix columns represent the sets, and each row corresponds to a set intersection. The cells encode whether the corresponding set is part of the corresponding intersection. 
In OnSet~\cite{sadana2014onset}, each cell of the matrix corresponds to an element. The matrix is copied for each set such that only the cells of elements in the set are colored. %
Frequency grids~\cite{micallef2012assessing} aim at communicating set sizes. Each entry of the matrix corresponds to an element, but set systems are usually small with few overlaps. %
Often the elements are placed set by set in horizontal~\cite{BROWN201192} or vertical order~\cite{price2007communicating}. Also a random order~\cite{brase2009pictorial} has been considered, where the contiguity of the sets has been sacrificed. As far as we know, there exists no work on the optimization of the elements' placement for frequency grids. In general, MosaicSets can be considered as a special form of a matrix-based technique, especially if a square grid is used. As OnSet and frequency grids, each element corresponds to a cell in the matrix. The novelty of our approach is the optimization of the mapping between elements and matrix cells such that sets form contiguous and compact regions also for non-trivial set systems.

A third class of set visualization techniques are overlays, where the placement of the elements is already given as input, e.g., from their spatial attributes. %
Overlay techniques have a limited scalability with respect to the number of elements and sets~\cite{DBLP:journals/cgf/AlsallakhMAHMR16}. 
An example of an overlay technique showing sets as regions is MapSets~\cite{ehkp-mvecg-15}, which partitions the underlying map into polygonal regions, so that each region contains only elements of one set. 
A generalization of MapSets are ClusterSets~\cite{geiger2021clustersets} that allow a set to be depicted by more than one polygonal region. 
In Bubble Sets~\cite{collins2009bubble} a smooth bubble-like region is computed for each set. The regions of different sets can overlap. 
Alternatively, overlay techniques can indicate sets by linear spanning structures. 
One example is LineSets~\cite{alper2011design}, which computes B\'ezier curves passing through the elements of each set using a traveling salesperson heuristic. 
A similar approach are Kelp diagrams~\cite{dinkla2012kelp} that visualize a set with a spanning graph structure. 
MosaicSets also uses an overlay component when visualizing the sets that do not belong to the base map (e.g., in our use case the research projects), yet the base map is computed such that these overlay sets have connected representations. %
We use existing techniques for displaying such overlay sets, e.g., by drawing their outer boundaries or using Kelp-like visualizations.

\subsection{Techniques from Geographic Information Visualization}\label{sec:relwork-geovis}
Although in our problem setting the data have no spatial component, several geovisualization techniques are visually and computationally similar. 
An approach that is visually similar to MosaicMaps are spatial treemaps, which partition a set of geometric elements according to a hierarchical tree structure.
Generally this does not lead to regular grid representations. An example for such spatial treemaps is an approach by Jern et al.~\cite{jern2009treemaps} that combines treemaps and choropleth maps. Efforts have been undertaken to tweak treemaps to produce regular grid visualizations, e.g., spatially ordered treemaps~\cite{wood2008spatially}, OD-maps~\cite{wood2010visualisation}, and spatial matrices~\cite{DBLP:journals/cartographica/WoodSD11}. Our approach differs in two aspects: (i) our elements do not need to be mapped according to a spatial location, (ii) our problem comprises contiguity requirements for the sets.

Visually similar to MosaicSets are also mosaic cartograms~\cite{DBLP:journals/cgf/CanoBCPSS15}. 
Here, we are given a political map with an integer weight per region. 
The aim is to  map each region to a set of grid cells in a square or hexagonal grid whose number is proportional to that region's weight.
In contrast to MosaicSets, mosaic cartograms shall represent the spatial input in terms of adjacencies, shapes and relative positions well and they only show a partition of the grid, but no overlapping sets.

Grid maps are an approach that is similar to MosaicSets both visually and in the problem setting. A grid map (also known as tile map~\cite{mcneill2017generating} or equal area unit map~\cite{schiewe2021distortion}) associates each area of a geographic subdivision of the plane, e.g., an administrative region, with a cell of a regular grid. 
Applications of grid maps are, e.g., in news and media~\cite{Berkowitz2015,Radburn2016} or in the visualization of research results~\cite{Slingsby2017,Kelly2013,DBLP:journals/cartographica/WoodSD11}. From a practical point of view the topic is addressed with a wide range of blog-posts, reports, and web-tools~\cite{Shaw2016,Wongsuphasawat2016,Zachary2015}. From a computational point of view a main challenge in computing a grid map is to find a suitable mapping between the geographic regions and the grid cells~\cite{schiewe2021distortion}. Eppstein et al.~\cite{DBLP:journals/ijcga/EppsteinKSS15} presented an approach that tackles this task considering three criteria: preserving (i)~location, (ii)~adjacencies, and (iii)~relative orientation. 
Meulemans et al.~\cite{meulemans2016small} presented an approach for generating grid maps for small multiples with a special focus on white space handling. 
McNeill and Hale~\cite{mcneill2017generating} published an approach where multiple grid maps are generated and the user can choose the most suitable one. 
Later a more complex problem setting was introduced where the geographic regions are partitioned into sets~\cite{DBLP:journals/tvcg/MeulemansSS21}. 
Now the matching between the geographic regions and cells has to satisfy that the union of cells of one set shall represent the corresponding geographic region well. 
Especially the latter version, where the subdivisions' elements are partitioned into sets, is similar to our problem setting. Nevertheless, our problem is different, since we can handle set systems with arbitrary set relations and we do not aim at providing a good geographic representation.

Another related technique %
was recently presented by Bekos et al.~\cite{Bekos2022}. They discuss the embedding of spatial hypergraphs, where each vertex has a spatial location and each hyperedge can join any number of vertices. They produce an embedding such that each vertex lies on a rectangular or concentric grid. As optimization criteria  the vertex displacement should be small and  the embedded edges should be visually clear. Their approach differs in several ways from MosaicSets: (i) their visualization leads to a sparse placement of vertices in a grid, while we aim at a compact representation using cells and (ii) they consider spatial closeness with respect to the given positions in their objective.

We want to emphasize that the challenges in geovisualization differ since the positions of the elements in the visualization should at least coarsely reflect their geographical locations. Hence, applying the presented techniques to our problem setting for non-spatial set systems would not lead to a satisfying solution. 

\subsection{Maps as a Metaphor}\label{sec:relwork-graphRepresentation}
As cartographic maps are well known, some visualization methods, including the previously discussed MapSets~\cite{ehkp-mvecg-15} and also MosaicSets, use this intuitive presentation of information to visualize non-spatial data; see the survey by Hogr\"afer et al.~\cite{hografer2020state}.
GMap~\cite{DBLP:conf/gd/GansnerHK09,DBLP:conf/recsys/GansnerHKV09} is an example using such a map metaphor that is in some sense similar to MosaicSets. Both approaches show non-spatial data. In GMap the input is a general graph and we are looking for a set of touching polygons that represent clusters in the data, which can be singletons or sets of closely connected nodes. Thereby the data and their underlying structure shall be represented in a comprehensible and visually appealing way. Other than our approach, GMap does not aim at a regular grid representation and their input is a graph and not a set system.

\subsection{Districting}\label{sec:relwork-districting}
The delineation of districts is a common problem in geographical information science and planning that is usually referred to as districting or zoning; for a detailed discussion we refer to \cite{RiosMercado2020}. 
Applications of districting include but are not limited to the definition of electoral districts \cite{validi2021imposing}, school districts \cite{caro2004school}, and ticket zones for public transportation systems \cite{tavares2007multiple}.
Usually, one aims at building the districts automatically on the basis of a given partition of the plane. For example, census tracts are grouped to electoral districts. 
Since MosaicSets groups grid cells to regions, we can adapt some ideas and techniques from existing districting methods.
However, MosaicSets
is not a districting method as
it also includes the assignment of elements to cells.

In the context of information visualization and cartography, districting methods have been proposed for map generalization tasks.
Oehrlein and Haunert~\cite{OehrleinH17} have presented a method for grouping the areas of a choropleth map to larger areas in order to obtain a less fine-grained choropleth map. Similar approaches have been used for land-use and land-cover maps \cite{haunertwolff2010b,GedickeCaGIS2021}. 
Although the constraints and objectives of districting depend on the concrete task at hand, some approaches have turned out to be rather generally applicable. This holds
for meta-heuristics such as simulated annealing \cite{ricca2008local} or evolutionary algorithms \cite{tavares2007multiple}.
These, however, usually require a feasible start solution, which is not easy to obtain in our situation.
Therefore, we choose an approach based on integer linear programming over meta-heuristics,
which has the additional advantage that it guarantees optimal solutions.
For MosaicSets we adapt
a model ensuring contiguous districts \cite{shirabe2009districting} and objectives favoring geometrically compact districts
from existing ILP formulations for districting.

A district is contiguous if every two points in it are connected via a curve that is entirely within that district. Contiguity is, thus, a qualitative property of a district.
Compactness, on the other, hand is a quantitative property. It can be expressed with various measures,
many of which are based on the area $A$ and the perimeter $P$ of a district. For example, the Polsby-Popper score of a district \cite{polsby1991third} is defined as 
\begin{equation}\label{eq:polsby-popper}
PP = 4\pi \frac{A}{P^2}
\end{equation}
and the Schwartzberg score \cite{schwartzberg1965reapportionment} as the square root of it.
Note that both scores evaluate to one for disks and attain values close to zero for highly non-compact shapes.  Furthermore, if the area $A$ is fixed,
maximizing a district's Polsby-Popper or Schwartzberg score  reduces to minimizing its perimeter.
Another general approach to measuring a district's compactness is based on eccentricities, i.e., distances between a center of a district and the areas assigned to it.
For MosaicSets we  use both perimeter- and eccentricity-based compactness measures.

\subsection{Support Graphs}\label{sec:relwork-supportgraphs}
Finally, we discuss on a more formal level the drawing of hypergraphs as mathematical models of set systems by means of suitable support graphs. 
Let $\mathcal H=(S,C)$ be a hypergraph with vertex set $S$ and hyperedge set $C$. A graph $G=(V,E)$ is called a \emph{support graph} (or simply \emph{support}) of $\mathcal H$ if $V=S$ and each hyperedge $c \in C$ induces a connected subgraph in $G$, i.e., the subset of edges in $E$ that connect pairs of vertices in $c$ is connected and spans all elements in $c$. 
Support graphs play an important role for visualizing hypergraphs (or set systems) since in a drawing of its support graph all hyperedges can be traced or highlighted as connected shapes, similarly to the regions representing a set in an Euler diagram.
Supports need to satisfy certain quality criteria in order to be considered useful for set visualization. 
Most prominently, a good support graph should be planar so that the shapes created by tracing the corresponding hyperedge subgraphs intersect only if their hyperedges share common vertices. 
Supports have been primarily studied from a theoretical perspective and it is known that deciding whether a given hypergraph admits a planar support is NP-complete~\cite{jp-hpcdvd-87}, even if the support must be 2-outerplanar~\cite{bkmsv-psh-11a}.
In contrast, testing whether a hypergraph has a support that is a path, cycle, tree, or cactus~\cite{ks-cmoct-03,jp-hpcdvd-87,bcps-bhaho-11,bkmsv-psh-11a} can all be done in polynomial time.
A practical approach based on computing hypergraph supports and rendering the visualization in the style of a schematic metro map with sets as transit lines is MetroSets~\cite{JacobsenWKN21}. 
While the above results concern supports of abstract hypergraphs, the problem of computing supports has also been studied for spatial hypergraphs, where the vertices have fixed positions in the plane and the support graphs need to preserve these vertex positions~\cite{cgmny-spssh-19}. 

MosaicSets does not assume fixed vertex positions, but it shares the use of support graphs for set visualization. 
There is, however, an important difference. 
We are using a predefined and bounded-size planar grid graph as a host graph into which we want to embed the given set system by means of finding a subgraph of the grid graph that serves as a support. 
This question has not been considered in the literature on hypergraph supports so far, which restricted only the class of support graphs, e.g., to trees, but not to graphs that are subgraphs of a specific host graph. 
One related result deals with the question if a given hypergraph has a tree support, where each vertex has a degree bound that must be respected by the computed tree support. 
Buchin et al.~\cite{bkmsv-psh-11a} proved that such a support can be computed in polynomial time (if one exists). 
However, even if our host grid implies a degree bound for each vertex, the grid-subgraph requirement is more restrictive (e.g., not every tree with maximum degree 4 can be embedded as a subgraph into a degree-4 square grid).
At the same time, we do not limit ourselves to tree supports, but consider the larger class of planar supports.

\section{Towards a Formalization of MosaicSets}\label{sec:formalization}

Let us recap the main principle of MosaicSets based on the example in Fig.~\ref{fig:sketch}.
The input is a set system, where each set is a department or a project and each element is a research group (or a key value representing it). The output is a map-like visualization where each element occupies a cell of a grid and each set is a contiguous region.
In this section, we first discuss the ideas that led to this principle as well as the design goals that we follow when generating the visualization (Sect.~\ref{sec:design})
and then present a formal definition of the problem  (Sect.~\ref{sec:problem}) as well as a short proof of its NP-completeness (Sect.~\ref{sec:hard}).

\subsection{Design Decisions} \label{sec:design}

The original intention behind MosaicSets 
is to visualize the structure of an institution (e.g., faculty) broken down   into smaller units (e.g., departments and, as atomic elements, research groups)
as well as intra-institutional collaborations.
The visualization should support the discussions at meetings of the institution's board and its commissions. 
In many situations, questions of strategic importance need to be discussed: What are the strengths and weaknesses of the faculty?
How large are the different departments, in terms of both absolute and relative sizes?
Between which parts of the faculty are the links well established and where are they underdeveloped? 
For instance, a frequent situation is that the chair of a recruitment commission presents the concept of a new professorship at a faculty board meeting. Then, a visualization should enable the commission chair to explain how the new position will be embedded within the faculty's organizational structure and research networks. 
Another application that we target is the communication of the faculty's structure to the outside world; e.g., the visualization may be published on the faculty's website or used in meetings with representatives from higher decision levels or funding agencies. 
The applications of MosaicSets are certainly not limited to academia, since, e.g., companies and non-profit organizations may be faced with similar strategic questions. 
Generally speaking, the above scenarios suggest creating a clean and intuitive visualization, which can easily be understood by the various stakeholders who are interested and/or already knowledgeable in the data, without a need for extensive training on how to read and use it. 
We aim for a visually attractive design that invites potential users to engage themselves with the visualization.
The visualization should enable the users both to get a quick overview of the data, but also to explore the relationships and patterns between and within the different entities of the depicted organization.
Based on this target use case, we emphasize that we do \emph{not} see MosaicSets as an exploratory or analytical tool for big data, which would require additional efforts for data simplification and aggregation.

\begin{figure}
    \centering
    \includegraphics[scale=.8]{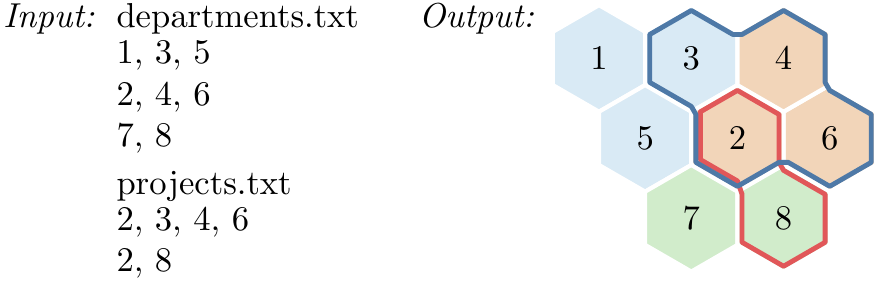}
    \caption{Required input files and a solution using a basic rendering style. Additionally, our method requires several parameter values as input.}
    \label{fig:sketch}
\end{figure}

The targeted scenarios require a visualization providing a strategic overview, while also being intuitive and engaging. For this reason, MosaicSets has been designed as a \emph{map-like visualization}: Just like the countries in a political map, each department occupies a dedicated territory.
A further essential design decision is to choose an \emph{equitable representation} of each research group: All research groups should be represented with an area of the same size. %
Requiring a map-like visualization and an equitable representation naturally leads to the use of a \emph{regular tessellation}, which directly entails the advantages of a \emph{schematic map}: low clutter and clear representation of topological relationships. Among the three existing types of regular tessellations, we favor \emph{hexagonal} or \emph{square tessellations} since their tiles are translated copies of each other, which does not hold for triangular tessellations as the third type. Moreover, in a triangular tessellation every tile has only three neighbors, which is a severe limitation when it comes to the embedding of set systems.
In square and hexagonal tessellations, on the other hand, there are four respectively six neighbors per tile. These tessellations are commonly used for maps in strategic board and computer games, which probably helps to understand the visualization as a metaphorical map.  

Next, we discuss the criteria that we use when assigning the elements of the sets to grid tiles.
Achieving that the sets are represented as \emph{contiguous regions} is our main goal, since research on human perceptual organization provides evidence that
``a connected region of uniform visual properties---such as
luminance or lightness, color, texture, motion, and possibly other properties as well---strongly tends to be organized as a single perceptual unit''
\cite{Palmer1994}.
(Note that the term ``principle of contiguity'' in psychology refers to a different concept, i.e., a close temporal relationship between stimuli and responses that leads to their association \cite{Lachnit2006}.)
Since the contiguity requirement can render the problem infeasible, we will consider problem variants with relaxed versions of it.
However, we will not relax the contiguity requirement for the sets corresponding to the organizational units that partition the institution of interest: For example, we strictly enforce that every department of a faculty is represented as a contiguous region, but we tolerate non-contiguous regions for projects in some problem variants.
This is because especially the departments should be displayed as territories to
give the visualization the look of a political map and because the
contiguity requirement can always be satisfied when embedding a partition (if a sufficiently large grid is used). This ``political map'' can then be used 
as the base map for visualizing the projects.
If the number of projects is small,
one may be able to visualize them all at the same time, e.g., by displaying the boundaries of different projects with different colors.
However, as this may quickly lead to visual clutter, we also consider displaying the base map with only one project at once.
For example, one can use \emph{small multiples} with different versions of the map showing different projects 
or \emph{interactive maps} allowing the user to select a single project or a small number of projects for display on top of the base map. Subject to the contiguity requirement, we aim in particular at maximizing the \emph{compactness} both of the visualization as a whole and of the region for each individual set,
since according to Gestalt theory humans tend to notice compact forms~\cite{Ehrenzweig1953} and since maximizing the compactness results in low visual clutter.

\subsection{Formal Problem Definition} \label{sec:problem}
As input we require a collection $C = \{S_1, \hdots, S_k\}$ of sets and an empty grid map in the form
of the adjacency graph $G = (V, E)$ of its cells.
We call $G$ the \emph{host graph} and observe that it is planar.
Let $S = \bigcup_{i=1}^k S_i$. 
We require $|V| \ge |S|$, i.e., the host graph must be large enough to accommodate all elements in $S$.
With $S$ and $C$, we obtain a hypergraph $\mathcal H = (S, C)$, where $S$ is the vertex set and  $C$ is the set of
hyperedges. 
Now, our basic task is as follows:
Find an injection $f\colon S \rightarrow V$ such that for $i = 1,\hdots,k$ the subgraph of $G$ induced by $\{f(s)\mid s \in S_i\}$ is connected. 
Note that $f$ being an injection implies that it maps each element in $S$ to exactly one node of $G$ and at most one element to each node; hence every element is assigned to one grid cell without any double occupancy.
Moreover, the subgraph  of $G$ induced by some $V'\subseteq V$ is $G' = G[V'] = (V',E')$ with $E' = \{\{u,v\} \in E \mid u,v \in V'\}$; 
hence requiring that $\{f(s)\mid s \in S_i\}$ induces a connected subgraph for each $S_i \in C$ is equivalent to the support graph property and hence implies that $G[V']$ for $V' = \{f(s) \mid s \in S\}$ is a planar support graph of $\mathcal H$. Furthermore, it
ensures that the region for each set $S_i$ is contiguous.

In order to tackle the basic task, we introduce a directed bipartite graph $B = (S \cup V, F)$ 
with $F = \{(s,v) \mid s \in S \wedge v \in V\}$ 
that models all possible assignments between elements and grid nodes as directed edges.
Selecting an assignment $(s,v) \in F$ thus means setting $f(s) = v$.
Finally, to differentiate between different feasible solutions, we introduce a generic cost $w(s,v) \ge 0$ for every possible assignment $(s,v) \in F$.
To express the compactness of the whole visualization as a basic optimization objective, we use 
\begin{equation}
w(s,v) = \left\| v.p - \mu \right\|\change{^2}\,, \label{eq:cost1}
\end{equation}
where $v.p \in\mathbb{R}^2$ is the position vector of the center of gravity of the cell represented by $v$ and $\mu\in\mathbb{R}^2$ the position vector of  the center of gravity of the set of all cells.

By minimizing the total cost with this setting, allocations of cells close to the center of the visualization are favored, yielding a compact visualization.
We will later refine this problem definition to also express that each individual set should be represented with a compact region and to express the compactness of a region based on its perimeter.

\subsection{Computational Complexity}\label{sec:hard}

We can prove that the problem of computing a support of an arbitrary hypergraph that is a subgraph of a planar host graph is NP-complete. 

\begin{theorem}
For a given planar graph $G=(V,E)$ and a hypergraph $\mathcal H = (S,C)$ with $|S| \ge |V|$ it is NP-complete to decide if there is a subgraph of $G$ that is a support graph of $\mathcal H$.
\end{theorem}

\begin{proof}
We reduce from the NP-complete \textsc{Planar Hamiltonian Circuit} problem~\cite{DBLP:journals/siamcomp/GareyJT76}, which is defined as follows. Given a planar, cubic, 3-connected graph $G=(V,E)$ determine if $G$ has a Hamiltonian circuit, i.e., a closed path that visits each vertex of $V$ exactly once and returns to its starting point.

It is easy to see that our decision problem belongs to the class NP as verifying whether a given subgraph of $G$ is actually a support of $\mathcal H$ can be done in polynomial time. For the hardness reduction, let $G=(V,E)$ be the graph for which we want to find a Hamiltonian circuit and define it as the planar host graph for our hypergraph support  problem. Further define the hypergraph $\mathcal H=(S,C)$ with $S=\{1,2,\dots,|V|\}$ and $C=\{ \{1,2\}, \{2,3\}, \dots, \{|V|-1,|V|\}, \{|V|,1\} \}$. Now if we can find a support of $\mathcal H$ in $G$, this support maps each element in $S$ to a unique vertex in $V$ and each hyperedge in $C$ to a unique edge in $E$; this support is thus necessarily a Hamiltonian circuit and thus exists if and only if $G$ is a Yes-instance of \textsc{Planar Hamiltonian Circuit}. 
\end{proof}

\section{An Approach Based on Integer Linear Programming} \label{sec:ilp}

In this section we present our approach using Integer Linear Programming, which is a general method for solving problems of combinatorial optimization.
Choosing this approach means that the solution to a problem of interest is encoded with an $n$-dimensional vector of integer variables ${\bf x} \in \mathbb{Z}^n$.
The aim is to minimize ${\bf c}^\intercal  {\bf x}$ for a given vector ${\bf c} \in \mathbb{R}^n$ while satisfying $A {\bf x} \succeq {\bf b}$ for a given matrix $A \in \mathbb{R}^{m \times n}$ and a given vector ${\bf b} \in \mathbb{R}^m$, where $\succeq$ refers to a row-wise comparison with $\ge$.
A problem encoded in this form is called integer linear program (ILP).  Although even the best known algorithms for solving ILPs  have an exponential worst-case running time, there exist sophisticated ILP solvers that often perform well in practice.
Hence, integer programming is a promising approach, especially for NP-hard problems such as our hypergraph embedding problem.
For more details about integer programming we refer to the textbook by Nemhauser and Wolsey~\cite{DBLP:books/daglib/0090563}.
Here, we focus on presenting the ILP for our problem.
First, in Sect.~\ref{ssec:basic_ilp}, we describe a basic ILP for finding a hypergraph support of the given set system within the given host graph.
We then present extensions to express the preference for compact regions (Sect.~\ref{ssec:compactness}),
a relaxed version of the contiguity requirement (Sect.~\ref{ssec:relaxation}), and a technique for reducing the number of variables of the ILP (Sect.~\ref{sec:better_ilp}).

\subsection{A Basic Integer Linear Program} \label{ssec:basic_ilp}

We first present the variables, the objective function, and the constraints of an ILP for embedding a set system (i.e., a hypergraph $\mathcal H = (S,C)$)
into a planar host graph $G=(V,E)$ while ensuring the connectivity of the sets. More precisely, solving the ILP yields a hypergraph support of $\mathcal H$ in $G$ if it exists.

\paragraph{Variables} We introduce one binary variable $x_{s,v} \in \{0,1\}$ for each potential assignment of an element $s \in S$ to a grid node $v \in V$ indicating whether it is selected ($x_{s,v}=1$) or not ($x_{s,v}=0$).
\begin{equation}
 x_{s,v} \in \{0,1\} \quad \forall (s,v) \in F
\end{equation} 
To ensure the contiguity of the sets, we adapt an approach by Shirabe~\cite{shirabe2009districting} for contiguity-constrained districting tasks. 
For this we introduce a directed version $\tilde{G} = (V, \tilde{E})$ of the undirected grid graph $G=(V,E)$,
by defining two opposite directed edges for each edge of $G$, i.e., $\tilde{E} = \{(u,v), (v,u) \mid \{u,v\} \in E\}$.
We use integer variables $y_{u,v}^{i}$ to model a multi-commodity flow in $\tilde{G}$.
Index $i$ refers to a set $S_i$ and $(u,v)$ to one directed edge of $\tilde{G}$.
\begin{equation}
y_{u,v}^{i} \in \{0,1,\hdots, |S_i| - 1\} \quad \forall i \in \{1,\hdots, k\}, \forall (u,v) \in \tilde{E}
\end{equation}
For each set $S_i$ there is one commodity. We arbitrarily select one element of set $S_i$ as the set's \emph{center}, which we denote as $c_i$. 
We will introduce constraints to ensure that
the cell to which $c_i$ is assigned serves as the sink of the flow of the commodity for $S_i$.
All other cells to which elements of $S_i$ are assigned serve as sources of that flow.
An example is shown in Fig.~\ref{fig:basic_ilp}.

\begin{figure}
\centering
\includegraphics[width=0.75\linewidth]{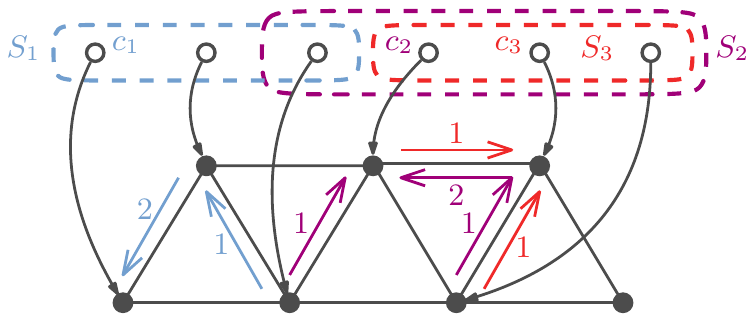}
\caption{An instance and a solution of the basic ILP.
Flows of different commodities are displayed with
arcs of different colors; they correspond to different sets. The arc labels show amounts of flow. The nodes to which a center (i.e., $c_1$, $c_2$, or $c_3$) is assigned serve as sinks.}
\label{fig:basic_ilp}
\end{figure}

\paragraph{Objective} The overall objective is to minimize the sum of costs for selected assignments using the generic cost function $w \colon S \times V \rightarrow \mathbb R_0^+$.
\begin{equation}
\text{Minimize} \quad \sum_{s \in S}\sum_{v \in V} w(s,v) \cdot x_{s,v} 
\label{eq:objective_basic}
\end{equation}

\paragraph{Constraints} We first introduce two constraints to ensure that the selected assignments satisfy the definition of an injection. 
\begin{equation}
\sum_{v \in V} x_{s,v} = 1 \quad \forall{s \in S} \label{eq:injection1}
\end{equation}
\begin{equation}
\sum_{s \in S} x_{s,v} \le 1 \quad \forall{v \in V} \label{eq:injection2}
\end{equation}
With the following two constraints we ensure the connectivity of each set based on the multi-commodity flow model.
\begin{equation}
\sum_{(v,w) \in \tilde{E}} y_{v,w}^i - \sum_{(u,v) \in \tilde{E}} y_{u,v}^i = \sum_{s \in S_i} x_{s,v} - |S_i| \cdot x_{c_i,v} \quad \forall{S_i \in C}, \forall{v \in V} \label{eq:flow1}
\end{equation}
\begin{equation}
\sum_{(u,v) \in \tilde{E}} y_{u,v}^i \le (|S_i|-1) \cdot \sum_{s \in S_i} x_{s,v} \quad \forall{S_i \in C}, \forall{v \in V} \label{eq:flow2}
\end{equation}
To demonstrate how the constraints for set $S_i$ control the flow of the corresponding commodity,
we discuss three different cases:
\begin{compactenum}[(1)]
\item If an element of $S_i$ other than the center $c_i$ is assigned to node $v$, the net-outflow from node $v$ (left-hand side of Constraint~(\ref{eq:flow1}))
is forced to 1 (right-hand side of Constraint~(\ref{eq:flow1})).
Moreover, the total inflow to node $v$ (left-hand side of Constraint~(\ref{eq:flow2})) is bounded from above by $|S_i|-1$ (right-hand side of Constraint~(\ref{eq:flow2})),
which is sufficiently large allowing $v$ to receive one unit from every other member of $S_i$.
\item If the center $c_i$ is assigned to $v$, the net-outflow from node $v$ is $1-|S_i|$ according to Constraint~(\ref{eq:flow1}),
which means that $v$ receives exactly one unit of flow for every other member of $S_i$.
\item If no element of $S_i$ is assigned to $v$, Constraint~(\ref{eq:flow2}) ensures that $v$ receives no incoming flow.
This and Constraint~(\ref{eq:flow1}) together imply that there is no outgoing flow from $v$.
\end{compactenum}
To summarize, the flow for $S_i$ can only enter or leave nodes to which elements of $S_i$ are assigned.
Since every such node (except the node to which $c_i$ is assigned) injects one unit of flow,
the connectivity requirement is satisfied.
We note that Constraints~(\ref{eq:flow1}) and (\ref{eq:flow2}) have been adapted in the following manner from Shirabe's flow model for districting tasks: The sum $\sum_{s \in S_i} x_{s,v}$ on the right-hand sides of the constraints replaces a single variable in Shirabe's model representing whether or not an area is selected for a certain district.

\subsection{Compactness of Regions Representing Sets}
\label{ssec:compactness}

Setting the cost function $w$ in  Objective~(\ref{eq:objective_basic}) as defined in Equation~(\ref{eq:cost1}),
the basic ILP produces a visualization that is compact as a whole.
However, this approach does not produce a visualization in which each individual set is represented as a geometrically compact region.
To achieve this, we propose two approaches:
Measuring compactness based on the regions' perimeters and based on eccentricities.  

\paragraph{Compactness based on perimeters}

As discussed in Sect.~\ref{sec:relwork-districting} the perimeter of a region is an appropriate measure of its compactness if the region's area is fixed.
This holds in our situation since the number of cells for each region is prescribed with the size of the corresponding set and since all cells have the same size.
To aggregate the compactness over all regions we simply minimize the sum of their perimeters.
We do this by giving a constant bonus for every set $S_i$ and every edge $\{u,v\}$ of $G$ that is within the region for $S_i$.
Maximizing the total bonus indeed corresponds to minimizing the sum of perimeters.
This is because every edge of $G$ corresponds to a boundary of constant length between two adjacent cells.
Including an edge in a region thus reduces the sum of perimeters by a constant amount.

We use additional binary variables to express the objective:
\begin{equation}
z_{u,v}^i \in \{0,1\} \quad \forall S_i \in C, \forall \{u,v\}\in E
\end{equation}
We force $z_{u,v}^i$
to zero if $\{u,v\}$ is \emph{not} within $S_i$, i.e., if to $u$ or to $v$ (or to both) no element of $S_i$ is assigned:
\begin{equation}
z_{u,v}^i \leq \sum_{s \in S_i} x_{s,u}\,,\qquad z_{u,v}^i \leq \sum_{s \in S_i} x_{s,v}\qquad\forall S_i\in C,\forall\{u,v\}\in E
\end{equation}
If both $u$ and $v$ are contained in the region for $S_i$, variable $z_{u,v}^i$ is free to receive value zero or one.
Hence, if we include it with a positive coefficient in a maximization objective, it will receive  value one whenever possible.
Accordingly, with the following objective, we minimize the total length of the boundaries of the regions representing sets:
\begin{equation}
\text{Maximize} \quad \sum_{S_i\in C}\sum_{\{u,v\}\in E} z_{u,v}^i \label{eq:obj-perimeter}
\end{equation}
In the experiments in which we optimized Objective~(\ref{eq:obj-perimeter})
we did not consider any other optimization objective, i.e., we replaced Objective~(\ref{eq:objective_basic}) with Objective~(\ref{eq:obj-perimeter}).
However, a combination of the two objectives would be possible by minimizing Objective~(\ref{eq:objective_basic}) minus Objective~(\ref{eq:obj-perimeter}), scaled by some constant factor expressing its importance.

\paragraph{Compactness based on eccentricities}

Eccentricity-based compactness measures usually charge costs proportional to distances or squared distances between a center of a district and centers of the mapping units contained in it.
Often, finding appropriate district centers is considered as an integral part of the districting problem at hand \cite{shirabe2009districting}. However, to speed up the computation, it is also common to prescribe the district centers, optimize the assignments of mapping units to the centers, and iterate with a new set of centers computed for the districts \cite{Hess1965}.
To implement this approach, we introduce $\mu_1, \hdots, \mu_k \in \mathbb{R}^2$ as the desired geometrical centers of the regions for sets $S_1, \hdots, S_k$.
We then apply the basic ILP as presented in Sect.~\ref{ssec:basic_ilp},
but instead of the assignment costs defined in Equation~(\ref{eq:cost1})
we use the following setting:
\begin{equation}
w(s,v) = \sum_{S_i \in C \colon s \in S_i} \left\| v.p - \mu_i \right\|\change{^2} \label{eq:cost2}
\end{equation}
This means that for every assignment $(s,v)$ we consider the sets containing $s$. For every such set $S_i$ we charge a cost proportional to the squared distance between the center of the allocated grid cell $v.p$ and the desired geometrical center $\mu_i$ of the region representing $S_i$.
Hence, it is favored to assign elements of $S_i$ to cells close to $\mu_i$.
Initially, we set $\mu_1 =  \hdots = \mu_k = \mu$,
meaning that the center $\mu$ of the entire grid is the desired geometrical center for every region.
We then compute a solution to the ILP with the assignment costs defined in Equation~(\ref{eq:cost2}). Afterwards, we compute the centers of gravity of the obtained regions and use these as $\mu_1, \hdots, \mu_k$ in a second run of the optimization algorithm. %
\change{This process is repeated multiple times until the centers $\mu_1, \hdots, \mu_k$ converge. In our examples, the changes of the center of each region become small after the third iteration, and fully converge after five iterations.}

When experimenting with the method %
we made an interesting observation: It took much longer to compute the first solution
(with the same center for every region) than the \change{subsequent solutions} %
(with different centers), even if in \change{all iterations} we performed a cold start of the solver.
A possible reason for this is that with the setting used in the first \change{iteration} %
there are many solutions of exactly the same quality. The solution in Fig.~\ref{fig:sketch}, for example, can be rotated or mirrored without changing the distances between the region centers and the center of the grid. This is not the case, however, with the setting used in the \change{subsequent iterations}%
, where the points  $\mu_1,\hdots, \mu_k$ differ from each other. 
This observation led us to the following approach: In the first \change{iteration}%
, instead of using $\mu$ as the center for every region, we draw a very small regular \change{$k'$-gon centered at $\mu$ with $k'$ referring to the number of base map sets. Each of the $k'$-gon's} corners \change{is set} as the center of one region \change{of the base map. The center of the overlay sets remains $\mu$.}
With this approach, the time needed for the first \change{iteration} %
decreased substantially.

\subsection{Relaxing the Contiguity Requirement}
\label{ssec:relaxation}

It is easy to imagine set systems that cannot be embedded into a hexagonal or square grid. Think, for example, of  seven sets, each of which contains two elements: one element that is shared by all sets and one element that occurs in no other set.
Moreover, even if an embedding exists, enforcing contiguity for too many sets can substantially decrease the quality of the solution with respect to the compactness objective. 

Therefore, if the task is to visualize many sets as overlays on top of the same base map, we suggest enforcing contiguity only for the regions of the base map. However, the compactness should be optimized for all sets (with respect to any choice of our compactness measures).
This approach is always feasible if a sufficiently large grid is used, since the sets of the base map constitute a partition.
Our current implementation renders any non-contiguous region as multiple regions, using the same color for their boundaries. \change{For an example, see a visualization of the Austrian parliament in the supplemental material.}
As a topic for future research, we also consider computing connected visual representations for non-contiguous regions, e.g., by overlaying linear connections between different parts of a region on top of the boundaries of the base map.

\subsection{An Integer Linear Program with Fewer Variables} \label{sec:better_ilp}

We greatly reduce the number of variables of our ILP by exploiting that, 
often, two distinct elements $s, t \in S$ occur in exactly the same sets.
More formally, we say that $s, t \in S$ are \emph{indistinguishable}  if 
\begin{equation}
s \in S_i \Leftrightarrow t \in S_i  \quad  \forall S_i \in C \,. \label{eq:contraction1}
\end{equation}
Note that when defining the assignment costs according to Equation~(\ref{eq:cost1}) or Equation~(\ref{eq:cost2}),
it holds for every two indistinguishable set elements $s$,$t$ and every cell $v$ that $w(s, v) = w(t, v)$.
Hence, given any solution, we can swap the cells of $s$ and $t$ and obtain a solution of the same quality.
This means that we can think of $s$ and $t$ as a single element that needs to be assigned to two grid cells.
\change{We note that carelessly removing such indistinguishable elements may change the existence of planar support graphs~\cite{bkkns-tsdh-16}.}
Hence we implement this idea by contracting \change{(but not removing)} indistinguishable set elements
in such a way that we obtain a smaller set system in which each element $s$ represents 
a number $\alpha(s)$ of original set elements.
This smaller set system has the same set relationships as the original one.
To assign $s$ to $\alpha(s)$ many cells, we replace the right-hand side of  Constraint~(\ref{eq:injection1}) with $\alpha(s)$.

When contracting elements, we need to take care that there remains at least one original element in each set that we can select as the set's center.
However, it is not too difficult to contract as many nodes as possible under this requirement.

\section{Rendering}\label{sec:rendering}

\begin{figure}
    \centering
    \includegraphics[scale=0.35]{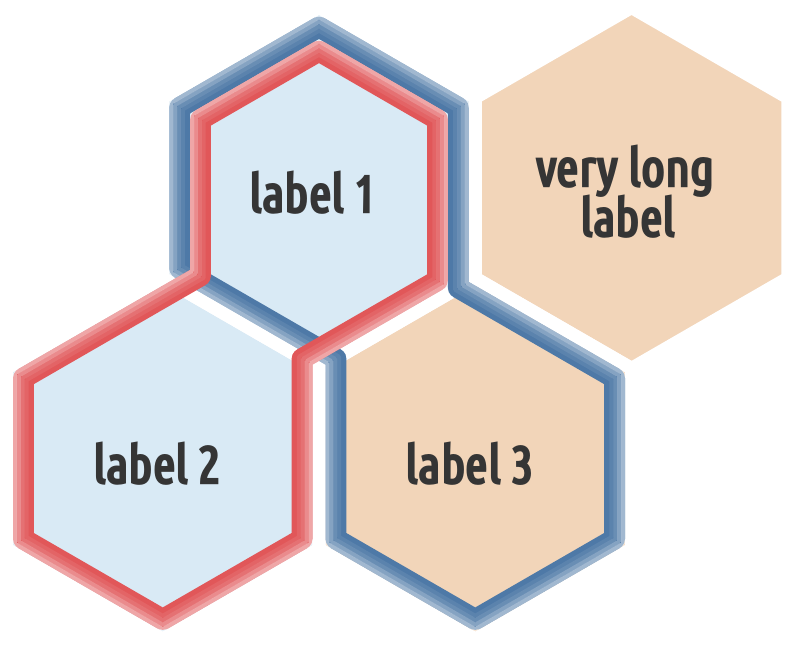}\qquad
    \includegraphics[scale=0.35]{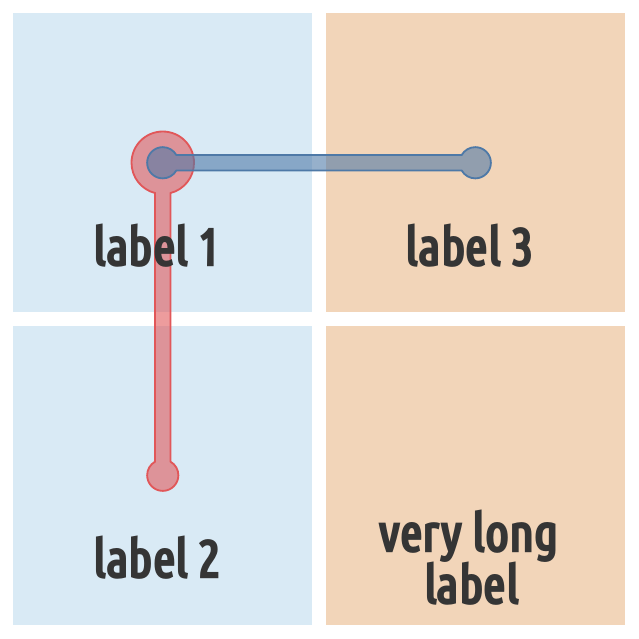}
    \caption{Different combinations of proposed rendering styles. %
    The left visualization shows a hexagonal tessellation with boundary style, the right visualization a tessellation with squares and Kelp style.}
    \label{fig:rendering}
\end{figure}

Our implementation of MosaicSets provides different rendering styles for the base map and set overlays as seen in Fig.~\ref{fig:rendering}. Firstly, we render the base map as a tessellation of either hexagonal or square grid cells. Additionally, we add spacing between adjacent cells to better distinguish between the inside and outside of the overlay regions. As computing the embedding requires a-priori information about the connectivity of the grid graph, the style of the base map must be declared in advance.

Secondly, we provide two different styles to render a set's region as an overlay on the base map. The first style renders the boundary of the region of the union of cells for each set. The boundary itself is drawn on the inner side of the boundary of the occupied cells. Potentially, boundaries of two or more sets in the same cell could overlap; we define an arbitrary order over all sets, draw the boundaries of sets according to the order but draw the boundary either next to an existing boundary or a cell's boundary. The thickness of the boundary can be specified in advance. \change{To better distinguish different overlay regions, we use a gradient coloring so that the brightness of a region's boundary decreases from the outside to the inside.} We refer to this style as \emph{boundary style}.

The second overlay style uses a similar style as seen in Kelp diagrams~\cite{dinkla2012kelp}. As we compute a flow in Sect.~\ref{sec:ilp} we can extract the subgraph of the flow network and use it to draw a Kelp-like overlay. Here, we represent the nodes of the subgraph as filled circles and use thick straight-line segments to represent the edges. Again, we have to define an arbitrary order over the sets first and process sets by said order as otherwise the drawn nodes and edges would overlap if two sets use either the same node or edge in the grid graph. By counting the number of sets already using certain edges and nodes in the grid graph, we can scale the diameter of circles and the thickness of straight-line segments such that all sets are visible. We refer to this as \emph{Kelp style}.    

Labels are either placed in the center of a cell or below the center if the Kelp style is used. The label size itself is scaled to fit the label in a cell of the base map. We set a maximum font size and check if all labels fit in their cell. If this is not the case we reduce the font size until either all labels fit or a minimum font size is reached. We automatically apply line breaks for white space and pre-defined delimiter characters, but only render a line break if it is necessary to fit the label in the cell.

Lastly, the color assignment of \change{the sets} is handled by using two palettes of different colors. This decision was guided by the idea that using different colors in each palette allows for a better differentiation between base map and overlays. For the base map we assign colors from a palette of light colors while for the overlays we use bright colors.

\section{Evaluation}\label{sec:evaluation}
In the following, we assess the quality of MosaicSets by discussing the opinions of experts and evaluating quantitative experiments. 

\subsection{Expert Interviews}\label{sec:eval-expertinterview}
\change{To evaluate MosaicSets, we conducted three separate interviews with two domain experts, who are administrative members of the Agricultural Faculty of the University of Bonn, and one designer who works for a design agency in the public transit sector. Prior to the interviews, the two administrative members collaborated to manually create a visualization of the research groups in the Agricultural Faculty of the University of Bonn, which actually inspired the idea of MosaicSets. We note that all visualizations showed in the interviews represent this running example. Combining expertise in the application domain and in graphic design, we consider the three interviewees suitable experts to evaluate our visualizations.} The questionnaires with all illustrations \change{and questions} can be found in the supplemental material. 

\begin{figure}
\centering
\includegraphics[width=0.8\linewidth]{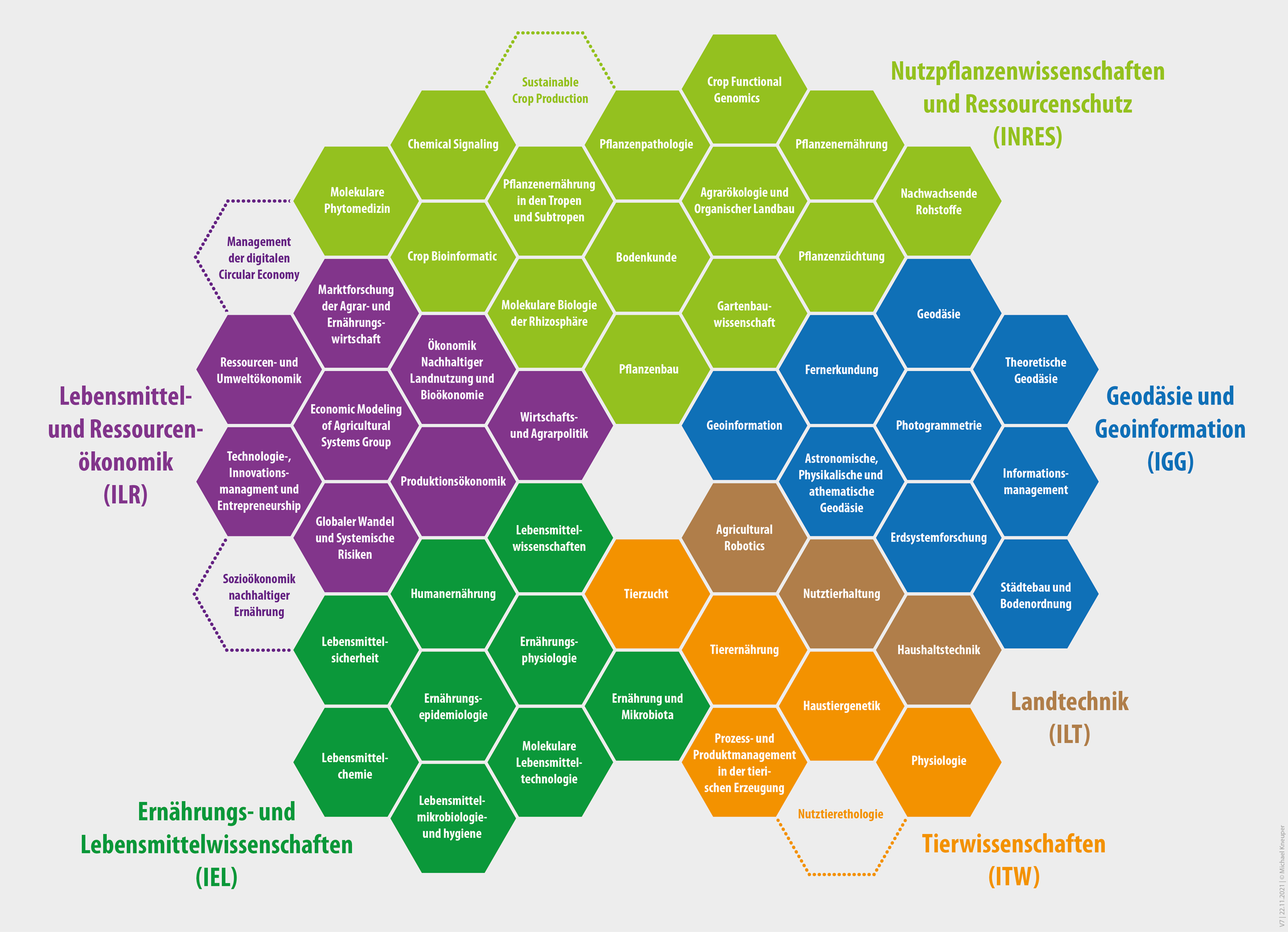}
\caption{Manually generated visualization showing the research groups of the Agricultural Faculty of the University of Bonn. Image credit to Dr.\ Susanne Plattes and Dipl.-Ing.\ Michael Kneuper. Figure~\ref{fig:teaser} illustrates MosaicSets of the same data set.}
\label{fig:expertSolution}
\end{figure}

First, we asked the experts to compare the manually generated visualization (see Fig.~\ref{fig:expertSolution}) with a visualization generated with MosaicSets. Both visualizations \change{show and optimize} only the base map. All experts emphasized as the primary difference the intentionally left blank cell in the middle of the manually generated representation where all departments are attached. \change{They stated that the empty cell helps to a create a common focus and to highlight the union of the faculty. However, one expert emphasized that the visualization generated with MosaicSets has less a hierarchical structure with respect to the departments' visualization and seems more flexible in terms of adding new data.} 

Next, we asked them to comment on \change{one MosaicSets visualization with hexagonal and one with square tiles}. They both represented the same sets consisting of six departments and two projects. \change{Although the experts perceived the square tiles as equally clear or slightly clearer than the hexagonal tiles, they emphasized that they generally found the hexagonal style more appealing.} They argued that the hexagonal tiles have more linking possibilities, which emphasizes the interrelationship of the departments more strongly. The entirety of the departments appeared as a unit in the hexagonal visualization. As an advantage of the visualization using square tiles, they mentioned the easier countability and identification of individual elements. The arrangement of the elements can easily be translated into rows and columns. Further, the experts pointed out that the visualizations differed with respect to the perceived hierarchy structure. While a hierarchy from the inside to the outside is perceived when hexagonal tiles are used, a hierarchy from the top to the bottom is implied when using square tiles. Overall, the experts saw advantages in both variants. For the running example, the hexagonal grid is more suitable since it better represents the union of the departments. \change{In the remaining part of the interviews we also focused on hexagonal visualizations.}

\change{We asked the experts to compare two versions of MosaicSets generated with different compactness measures. The first version is computed with compactness based on eccentricity. The second is a variant of compactness based on eccentricity that restricts the available grid cells to those used in the first iteration and therefore fixes the overall shape of the grid. In particular, in the first version the focus is on the compactness of each set and in the second version the overall map is more compact. The experts did not show a clear preference. While it is easier to differentiate the individual sets in the first version, the overall appearance is more homogeneous in the second version. The experts agreed that which version to choose depends on the application scenario. } 

Comparing MosaicSets with two, three and four overlaying projects, the experts pointed out a problem with visual clutter and separability when there are more than three superimposed projects. Especially, the distinction between inner and outer segments becomes more difficult. Hence, the experts suggested the use of small multiples showing the same base map but different projects. 

We also provided the experts with an interactive map allowing the display of only a single project or a chosen selection of projects (see the web page referenced at the end of Sect.~1). When the representation of many projects is important, the experts saw great merits in the interactive map. The experts suggested to add further interactions as highlighting of only the tiles belonging to the selected projects. \change{Further, they proposed adding alternative data representations, for example tables stating  cardinality, next to our visualizations.} We think these suggestions are well suited to be incorporated into our visualization. 

\change{We also asked the experts for general feedback and they} indicated that the proximity of research groups with joint interests or publications could be another criterion worth considering. Further, they would like to see a consistency criterion that maintains the basic arrangement of the departments' regions \change{when additional institutes, projects, or working groups are added or removed}. 

\change{In preliminary interviews we asked two of the experts on their opinion with respect to the rendering.} Regarding the overlay style, the experts agreed that boundary style is more suitable than Kelp style. Using Kelp style clarity is lost as the selection of tiles connected with segments is not intuitive. When using boundary style, the experts emphasized that choosing a high-contrast color scheme is important to ensure quick perception of the projects. 

Overall, the experts considered MosaicSets to be a valuable approach for visualizing set systems. In particular, the interactive visualization meets the experts' qualitative requirements. They pointed out that there is a high demand for such visualizations in a wide range of application areas, e.g., for internal sessions and meetings as well as for external exposure (e.g., via the faculty website). Due to the highly dynamic structure of the faculty (e.g., research projects starting or expiring) updated visualizations have to be generated frequently. The discussion with the experts revealed that the manual generation of a visualization is highly time-consuming. To create the visualization shown in Fig.~\ref{fig:expertSolution}, the corresponding two experts invested several working hours, a large part of which was needed to arrange the tiles. Hence, using MosaicSets can save a \change{substantial} amount of valuable working time.

\subsection{Tasks for Set Visualizations}\label{sec:eval-tasks}
In the following, we aim to assess MosaicSets with respect to the task taxonomy defined by Alsallakh et al.~\cite{DBLP:journals/cgf/AlsallakhMAHMR16}. The tasks are summarized into three groups:~(A)~tasks related to elements;~(B)~tasks related to sets and set relations; and~(C)~tasks related to element attributes. For some of these tasks we can directly state that MosaicSets is not able to solve them. We cannot solve the tasks of type~(C) since we do not consider any element attributes in MosaicSets. Further, tasks~A5--A7, B13 and B14 require an interaction technique that MosaicSets does not provide and task~B11 asks for set similarity measures that are not defined in our problem setting.
To assess which of the remaining 15 tasks can be solved with MosaicSets, we asked the experts from our interviews (see Sect.~\ref{sec:eval-expertinterview}) to \change{solve one example for each task} and to rate whether a task is 'fully', 'partially', or 'not' supported  (see the supplemental material for the full task list and the experts' ratings). We note that the experts' assessment was based on a static visualization with six departments and two projects. \change{All tasks were answered correctly by all experts.} The tasks A1, A3, B1--B10, B12 are considered fully supported \change{and the experts were able to solve them within a few seconds}. For tasks A2 and A4 the experts gave different answers; each twice 'fully supported' and once 'partially supported'. \change{The expert who rated the tasks as 'partially supported' took much longer to complete them ($>30$\,s).} 

In summary, we claim that almost all element- and set-based tasks except A5-A7, B11, B13, and B14 are supported by MosaicSets. Compared to other methods from the literature, MosaicSets performs like a combination of Euler-diagrams and frequency grids. In detail, for (A)-tasks MosaicSets performs better than Euler-diagrams, which do not support two (A)-tasks, and it performs nearly as good as frequency grids that support all (A)-tasks. For (B)-tasks MosaicSets performs better than frequency grids but not as good as Euler diagrams. Euler diagrams support all except B13, while frequency grids on the other hand support all (B)-tasks except B3, B4, B11, and B14.

\newcommand{\databonn}{\textsc{Bonn}}
\newcommand{\datavienna}{\textsc{Vienna}}
\newcommand{\datanationalrat}{\textsc{Parliament}}
\newcommand{\MosaicSetsP}{\textsc{MSP}}
\newcommand{\MosaicSetsE}{\textsc{MSE}}
\newcommand{\MosaicSetsEArea}{\textsc{MSEA}}

\subsection{Experimental Setup}
\label{ssec:setup}
We quantitatively compare different versions of MosaicSets: \MosaicSetsP{} with compactness based on \textbf{p}erimeters, \MosaicSetsE{} with compactness based on \textbf{e}ccentricities optimized \change{in multiple iterations}, \change{and \MosaicSetsEArea{}, which is a  variant of \MosaicSetsE{} that restricts the used grid cells to the assigned grid cells of the first iteration and therefore fixes the assigned \textbf{a}rea}. With all three versions we enforce all sets to form contiguous regions. 
For all versions we used the elimination of variables presented in Sect.~\ref{sec:better_ilp}, as it reduces the running time by an order of magnitude.

We used three data sets: \databonn{} consists of 51 unique elements and 9 sets, with 6 sets used as base map; \datavienna{} consists of 7 sets with 71 unique elements and 4 sets used as base map; and \datanationalrat{} consists of 8 sets with 178 unique elements, where 5 sets are used as base map. \databonn{} refers to the Agricultural Faculty of the University of Bonn and \datavienna{} to the Faculty of Informatics of TU Vienna. \datanationalrat{} refers to the Austrian parliament with parties as sets of the base map and interest groups as overlays. The host graph has as many rows as columns and is adapted to the number of unique set elements. We add one additional row and column to the minimum required number. 

We performed the experiments on an \change{AMD Ryzen 7 PRO 4750U with $16$\,GB} of memory, implemented in Java and used the ILP solver of Gurobi 9.5.1. %
\change{For all solutions we used a maximum optimality gap in Gurobi of $0.5\%$.} %
\change{Additional} visualizations are provided in the supplemental material.

\savebox{\mybox}{\includegraphics[scale=0.18]{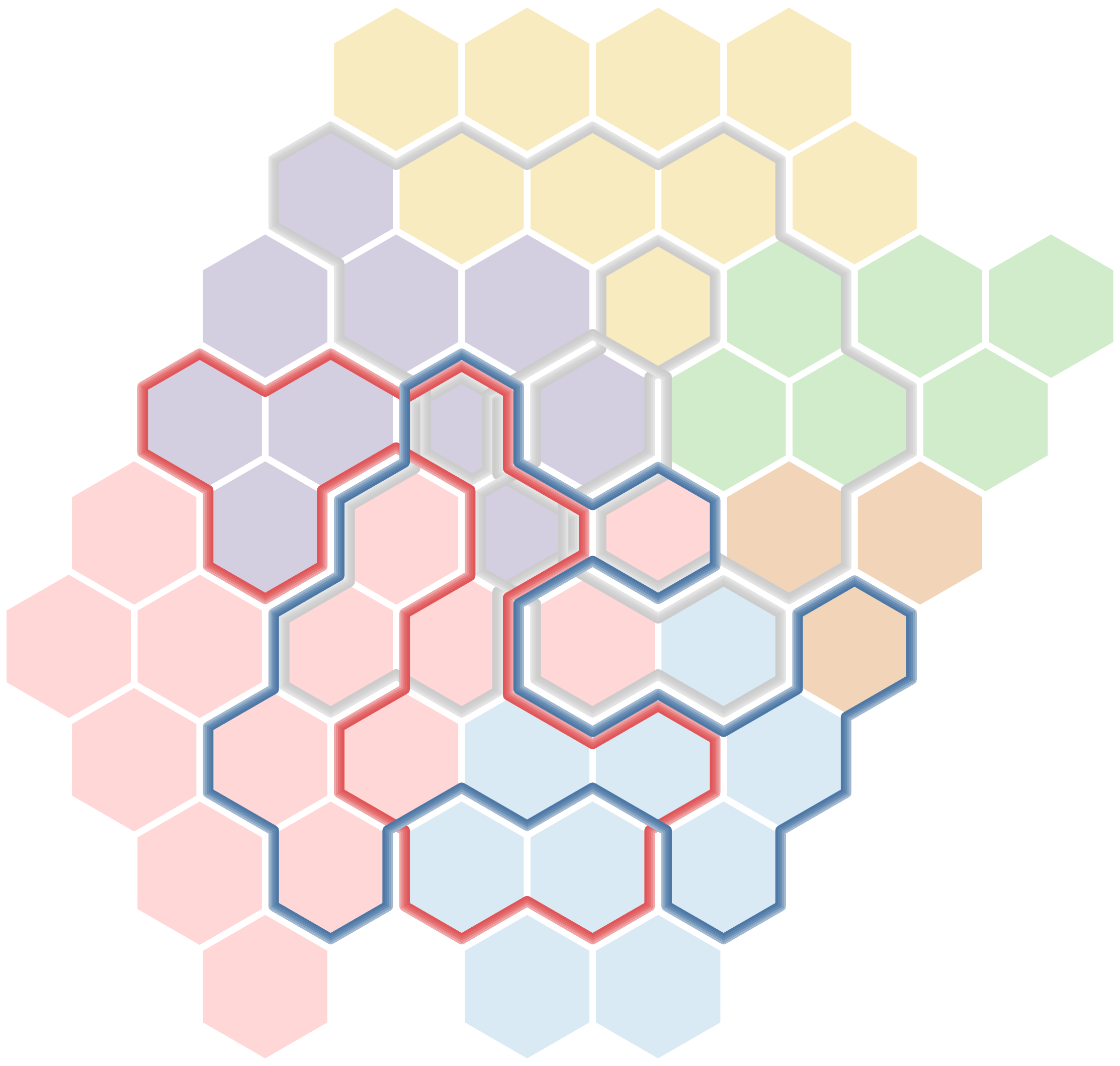}}
\begin{figure}[t]
\centering
\begin{subfigure}[c]{0.4\linewidth}
        \centering
        \vbox to \ht\mybox{%
        \vfill
        \includegraphics[scale=0.18]{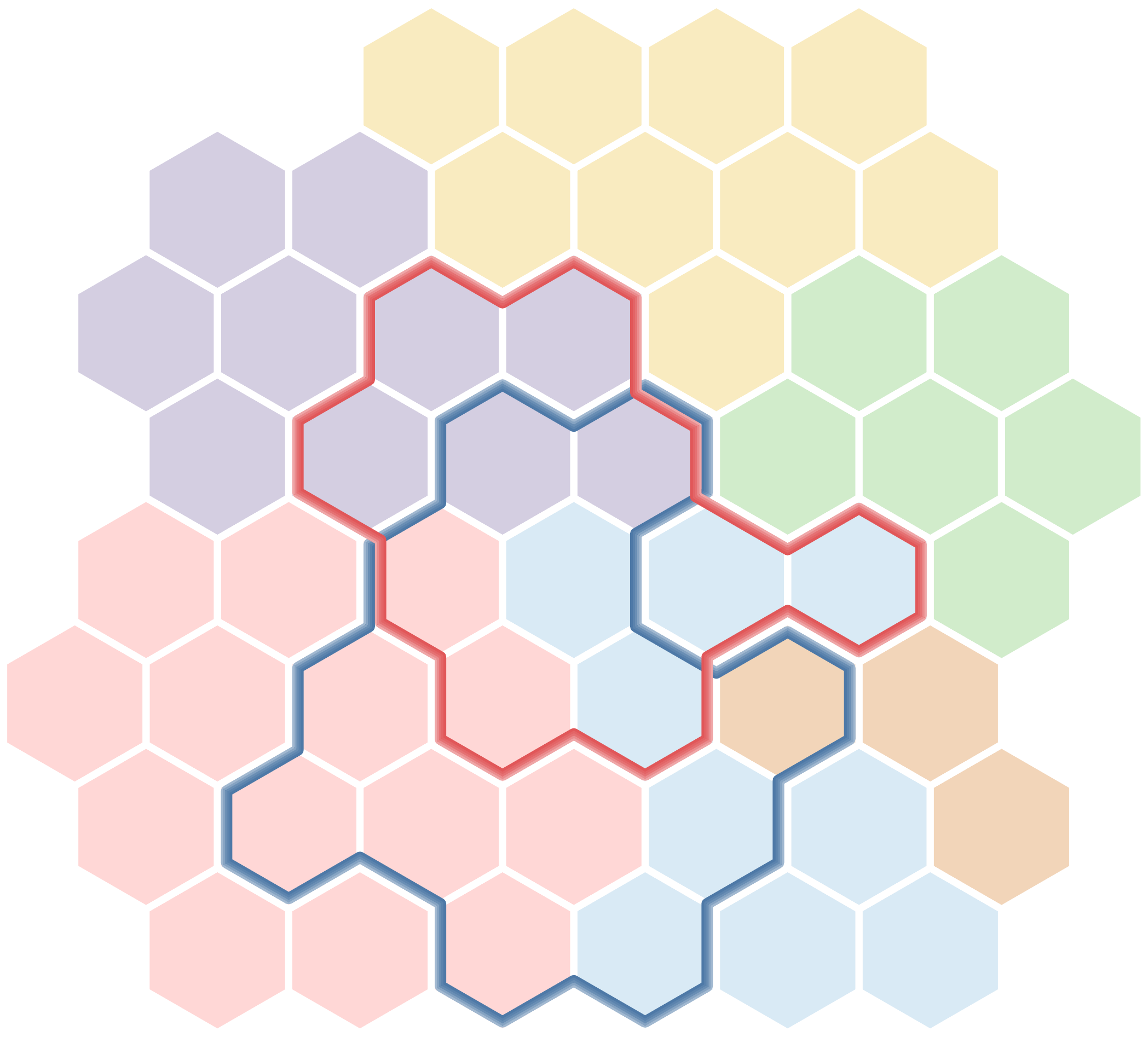}
        \vfill
        }
        \caption{\change{two projects}}
        \label{fig:lwf2proj}
\end{subfigure}
\qquad
\begin{subfigure}[c]{0.4\linewidth}
        \centering
        \usebox{\mybox}
        \caption{\change{five projects}}
        \label{fig:lwf5proj}
\end{subfigure}
    \caption{\change{MosaicSets computed with eccentricity-based compactness (\MosaicSetsE{}) for \databonn{} with different numbers of projects. We highlight the same two projects (red and blue) and indicate the others in gray.}}
    \label{fig:manyProjects}
\end{figure}

\begin{figure*}[t]
    \centering
\begin{subfigure}{0.28\linewidth}
    \centering
    \includegraphics[width=\linewidth,page=1]{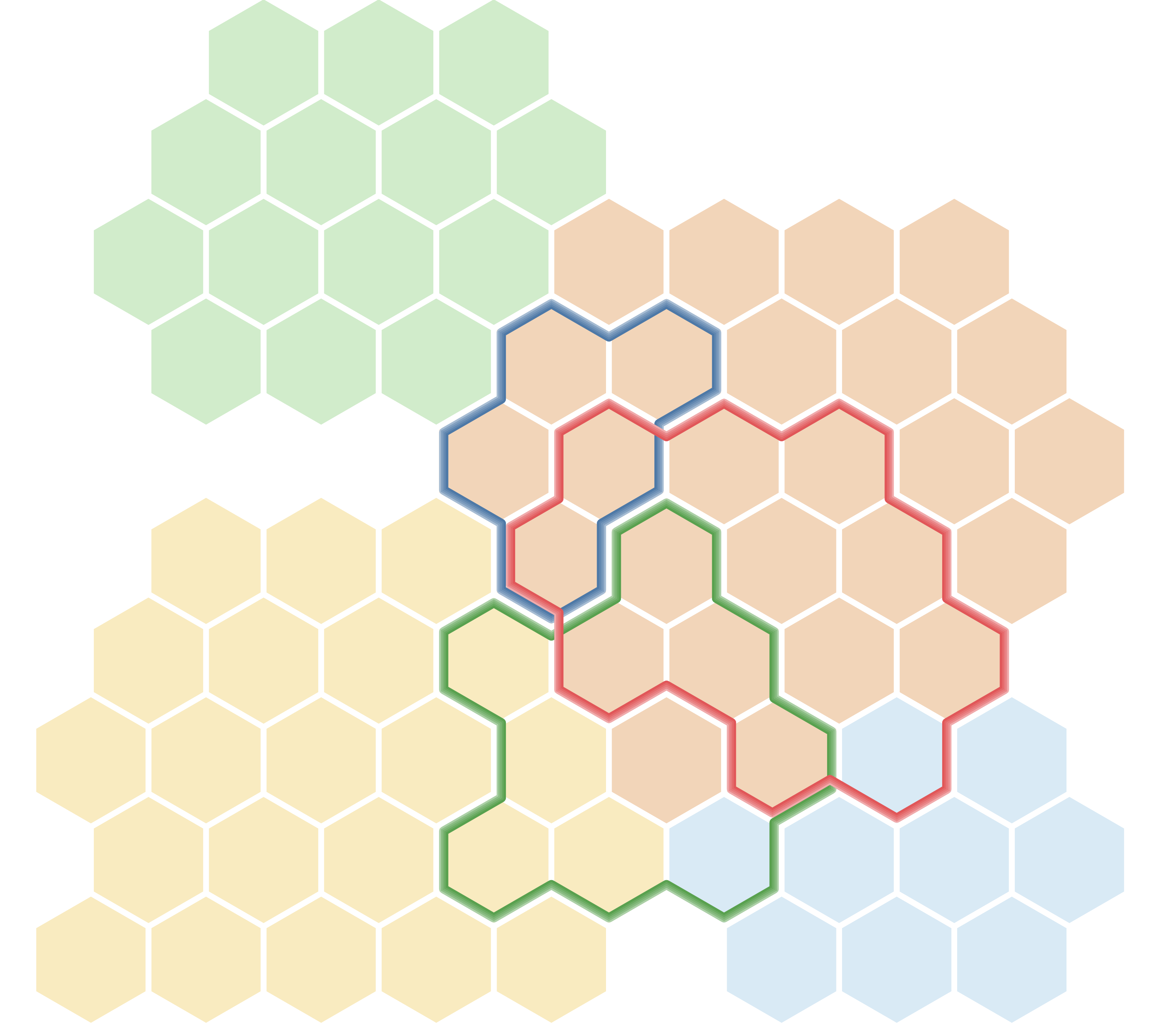}
    \caption{perimeter-based compactness (\MosaicSetsP)}
    \label{fig:perimeter}
\end{subfigure}
\hspace{2em}
\begin{subfigure}{0.28\linewidth}
    \centering
    \includegraphics[width=\linewidth,page=2]{tuw_fig7}
    \caption{eccentricity-based compactness (\MosaicSetsE{})}
    \label{fig:eccentricities}
\end{subfigure}
\hspace{2em}
\begin{subfigure}{0.28\linewidth}
    \centering
    \includegraphics[width=\linewidth,page=3]{figures/tuw_fig7.pdf}
    \caption{as (\subref{fig:eccentricities}) but area fixed after first iteration (\MosaicSetsEArea{})}
    \label{fig:instOnly}
\end{subfigure}
\caption{Comparison of the different compactness approaches for \datavienna{} \change{using three overlay sets}.}
\label{fig:comparison}
\end{figure*}

\subsection{Number of Overlay Sets}\label{sec:eval-numOverlays}

From the expert interviews (Sect.~\ref{sec:eval-expertinterview}), we learned that overlaying more than three sets (e.g., research projects) severely limits visual clarity. In principle, this problem can be counteracted by using small multiples or interactive maps. For both strategies it is necessary to first generate a visualization containing all sets of interest. 
\change{Figure~\ref{fig:manyProjects} shows \databonn{} optimized once with two and once with five overlay sets but both times only two projects are visualized. For an interactive version of Figure~\ref{fig:lwf5proj} that allows to display all five overlays see link in Sect.~1}
When applying \MosaicSetsE{}, the number of overlay sets strongly affects the running time, \change{e.g., from $0.8$\,s to $40.8$\,s in Fig.~\ref{fig:manyProjects}. Computing a solution with four overlay sets took $2.8$\,s.}
Moreover, we observe that as the number of considered overlay sets in \MosaicSetsE{} increases, the compactness of the sets decreases.
\change{While the base map in Fig.~\ref{fig:manyProjects} is comparably compact in both visualizations, the overlay sets are less compact with five overlays in Fig.~\ref{fig:lwf5proj}. For the two highlighted projects, the Polsby-Popper score decreases from $0.486$ to $0.234$ and from $0.500$ to $0.236$, respectively.} 

\change{We recommend using small multiples or interactive maps for a large number of overlay sets. Otherwise, in case of static maps, using a large number of overlay sets results in overloaded and cluttered maps. For the further experiments we focus on static maps and hence limit to three overlay sets which complies with the experts' recommendation.}

\subsection{Running Time}
\change{To investigate the influence of the two different compactness formulations (see Section~\ref{ssec:compactness}) on the running time, we compare the running times of the variants \MosaicSetsE{}, \MosaicSetsP{} and \MosaicSetsEArea{}. We evaluate them for \datavienna{} with three projects and a hexagonal grid. While both \MosaicSetsE{} and \MosaicSetsEArea{} were solved in less than one second, we stopped the computation of \MosaicSetsP{} after one hour with an optimality gap of \change{$15.69\%$}. Thus, we argue that in terms of a reasonable running time, compactness based on eccentricities is a better choice. Nevertheless, we suppose that \MosaicSetsP{} is well suited to assess the compactness of \MosaicSetsE{} and \MosaicSetsEArea{}.}

\change{Next we focus in more detail on the running times of both versions implementing compactness based on eccentricities: \MosaicSetsE{} and \MosaicSetsEArea{}. Across all three data sets, the first iteration of \MosaicSetsE{} needed a maximum of $0.7$\,s. We observe that the running time of all subsequent iterations is substantially faster. Terminating after a total of five iterations, the combined computation time of the four subsequent iterations is of the same order of magnitude as for the first iteration. Overall, the computation time for a hexagonal grid is $0.7$\,--\,$1.6$\,s and for a square grid $0.6$\,--\,$1.6$\,s.} %
\change{Considering \MosaicSetsEArea{}, we do not observe a substantial difference in the running time compared to \MosaicSetsE{} for \databonn{} and \datavienna{}. However, solving the largest of our data sets, i.e., \datanationalrat{}, it is about $10\%$ slower. This holds for both the hexagonal and square grid.}

\subsection{Assessing the Compactness}\label{sec:eval-compactness}
When assessing the compactness of our visualizations, we deem three objectives to be desirable: (C1) the overall base map is compact, (C2) each individual set is compact, and (C3) only base map sets (e.g., departments) are compact. To quantify the compactness, we use the Polsby-Popper score which we denote by $PP_{\textrm{C1}}$ for (C1). For (C2) and (C3) we compute the mean Polsby-Popper score over the considered sets which we denote by $PP_{\textrm{C2}}$ and $PP_{\textrm{C3}}$, respectively. 

In Fig.~\ref{fig:comparison} we show the results for \MosaicSetsP{}, \MosaicSetsE{}, and \MosaicSetsEArea{} for \datavienna{} and a hexagonal grid.
We observe that for \MosaicSetsP{} the overall base map (C1) is less compact compared to the ones achieved by \MosaicSetsE{} and \MosaicSetsEArea{} which is reflected in $PP_{\textrm{C1}} = 0.381$, $PP_{\textrm{C1}} = \change{0.532}$ and $PP_{\textrm{C1}} = \change{0.603}$, respectively. \change{This can also be observed when considering Fig.~\ref{fig:comparison}, where \MosaicSetsEArea{} results in the most compact overall base map.}
Considering criterion (C2), we achieve a better result for \MosaicSetsP{} ($PP_{\textrm{C2}} = \change{0.618}$) compared to \MosaicSetsE{} ($PP_{\textrm{C2}} = \change{0.574}$) and \MosaicSetsEArea{} ($PP_{\textrm{C2}} = \change{0.576}$).
For criterion (C3), we achieve the best result for \MosaicSetsP{} ($PP_{\textrm{C3}} = \change{0.640}$). With \MosaicSetsE{} ($PP_{\textrm{C3}} = \change{0.586}$) we achieve a slightly lower value regarding (C3). However, fixing the area of the base map after the first iteration with \MosaicSetsEArea{}, we obtain the lowest compactness of these sets ($PP_{\textrm{C3}} = \change{0.569}$).
\change{This reduced compactness was also observed by the experts when they looked at the corresponding visualizations.}
\change{Overall, we argue that \MosaicSetsE{} provides a compromise between all criteria while still being applicable in practice considering the running time. With respect to the experts' opinions (see Sect.~\ref{sec:eval-expertinterview}), we suggest \MosaicSetsE{}, since both the compactness of the entire base map (C1) and of the sets building the base map (C3) were considered particularly important.}

\section{Conclusions and Future Research}
We have presented an approach for the visualization of set systems on a regular grid. We require that all sets form contiguous regions and look at two different models for the compactness. In particular, we have a measure based on the perimeter of the region and one based on eccentricities. As the underlying problem is shown to be NP-hard, we presented ILP formulations of the two variants of the problem. 

To evaluate MosaicSets, we interviewed experts about visual criteria and applicability in a real-world scenario. The experts pointed out that MosaicSets is visually as good as a manually generated illustration, while saving several days of labour. Further, we performed experiments with three real-world data sets. In the experiments, we compared the perimeter and eccentricities variants with respect to the compactness; it shows that both have their strengths and weaknesses but none outplays the other. On the other hand, when comparing the running times it shows that MosaicSets with perimeter compactness runs in more than an hour while MosaicSets with eccentricities compactness can be computed within few seconds. Hence, we recommend using eccentricities compactness and consider this variant to be applicable in practice.

Still, we see a large potential for future research on MosaicSets. 
For example, the experts pointed out that combining MosaicSets with other information visualization techniques (e.g., tables) and interaction techniques (e.g., highlighting of cells) \change{can provide a more refined and powerful visualization system.}
Considering the rendering, we could include the optimization of colors (similar as in GMap~\cite{DBLP:conf/gd/GansnerHK09,DBLP:conf/recsys/GansnerHKV09}) and \change{optimize routing of region boundaries. Further, a comprehensive comparison and support of different rendering styles such as KelpFusion~\cite{meulemans2013kelpfusion} could give users more flexibility for different use cases.} It would be beneficial to develop a user interface for the generation of MosaicSets that includes interactive model manipulation. For example, the user should be able to introduce use-case-specific constraints such as the constraint asking to place each of six departments adjacent to a central empty hexagonal grid cell \change{or specify other global or local shape constraints}.  %
Also dynamically evolving and temporal set systems lead to open research questions, e.g., when considering consistency criteria. Finally, our work brings up interesting new questions for the algorithms community: In which cases (i.e., for which classes of hypergraphs and graphs) can we compute a hypergraph support in a given host graph in polynomial time? Are there efficient approximation algorithms for practically relevant versions? It would be interesting to compare such algorithms with our ILP-based approach.

\acknowledgments{We thank the experts from our study Michael Kneuper and Dr. Susanne Plattes (University of Bonn, Germany), and Lara Orth (die Informationsdesigner, Cologne, Germany). Partially funded by the Deutsche Forschungsgemeinschaft (DFG, German Research Foundation) under Germany’s Excellence Strategy – EXC 2070 – 390732324 and by the Vienna Science and Technology Fund (WWTF) grant ICT19-035.}

\bibliographystyle{abbrv-doi-hyperref}

\bibliography{references}
\end{document}


\firstsection{}
\setcounter{figure}{7}
\maketitle
\appendix

\newcommand{\databonn}{\textsc{Bonn}}
\newcommand{\datavienna}{\textsc{Vienna}}
\newcommand{\datanationalrat}{\textsc{Parliament}}

\newcommand{\MosaicSetsE}{\textsc{MSE}}
\newcommand{\MosaicSetsEA}{\textsc{MSEA}}
\newcommand{\change}[1]{\textcolor{black}{#1}}
\newcommand{\pp}{PP_{C_2}}

\change{
In this document, we provide supplemental material that is not included in the main manuscript due to space limitations. We provide additional visualizations in Sect.~\ref{apx:figs}, additional information on the tasks completed by the experts in Sect.~\ref{apx:tasks}, and the manuscript of our expert interviews in Sect.~\ref{apx:script}. 
}
\section{Additional Figures}\label{apx:figs}
\change{
\paragraph{Figures for Section~4.3}
Figure~\ref{fig:relaxation} shows the proposed relaxation of the contiguity constraint which is explained in Sect.n~4.3. In detail, we do not enforce the contiguity of the overlay sets anymore. We use \datanationalrat{} for the comparison where we include eight overlay sets. When trying to solve this with the eccentricity-based approach \MosaicSetsE{}, we terminated the experiment after ten minutes. Until then \MosaicSetsE{} did not find a valid solution. Relaxing the contiguity constraint led to a solution within few seconds. The caption gives detailed running times.}
\change{
\paragraph{Figures for Section 5}
In Fig.~\ref{fig:kelp_style}, we illustrate examples for MosaicSets using the Kelp rendering style. We computed MosaicSets with \MosaicSetsE{}. The corresponding boundary style visualizations are given in this document in Fig.~\ref{fig:V_hx}a and  Fig.~\ref{fig:parliament_sq}a, respectively.}
\change{
\paragraph{Figures for Section 6.5 and Section 6.6.}
In Fig.~\ref{fig:bn_hx} to Fig.~\ref{fig:parliament_sq}, we compare visualizations of MosaicSets computed with \MosaicSetsE{} and \MosaicSetsEA{}. We recall that \MosaicSetsE{} uses an eccentricity-based compactness measure in all iterations. \MosaicSetsEA{} is a variant of \MosaicSetsE{} that restricts the available grid cells to those used in the first iteration and therefore confines the available grid area. In the captions of the figures, we also provide measures on compactness (mean Polsby-Popper score) and the running times. We show such a comparison for each data set with a hexagonal grid and afterwards with a square grid.}

\change{\databonn{}} is shown in Fig.~\ref{fig:bn_hx} and Fig.~\ref{fig:bn_sq}. The data set consists of \change{51} research groups,  \change{six departments} and three research projects. We use the departments as sets of the base map and visualize the research projects as overlays. 

\change{\datavienna{}} is shown in Fig.~\ref{fig:V_hx} and Fig.~\ref{fig:V_sq}. The data set consists of 71 research groups belonging to four different institutes and we have three research projects. Again the institutes are used as sets for the base map and the research projects are visualized as overlays.

\change{\datanationalrat{} is shown in Fig.~\ref{fig:parliament_hx} and Fig.~\ref{fig:parliament_sq}. The base map represents all 178 members of the Austrian parliament colored by their affiliation to one of the five political parties. The coloring of the base map uses lighter shades of the different parties' official colors. The overlay shows interest groups. Such a map could aid in transparency by understanding if a vote was cast to serve the interest of groups, networks, or individuals.}

\section{Task Taxonomy}\label{apx:tasks}
\change{Alsallakh et al.~\cite{DBLP:journals/cgf/AlsallakhMAHMR16} introduced a taxonomy of tasks commonly associated with set visualization techniques. Tasks are classified into three broad categories:}

\begin{compactenum}
    \item Tasks in Group \textbf{A} are element-based tasks that are concerned with elements and their respective relationship to the sets. For example: In which research projects is the research group 'Data Science in Agricultural Economics' involved in?
    \item Tasks in Group \textbf{B} are related to sets and the relationship between different sets without taking individual elements into account. For example: Which research projects do overlap with project 'PhenoRob'?
    \item Tasks in Group \textbf{C} are related to element attributes and consider attributes of set elements and their relationship of distribution in regards to set membership. For example: Do research groups in the 'PhenoRob' project publish more papers than research groups in the 'DETECT' project? 
\end{compactenum}

MosaicSets only supports tasks in Groups A and B as we do not represent element attributes. Table~\ref{tab:taxonomy} shows all tasks of the two categories A and B. We list the assessment from our experts whether MosaicSets supports a task fully, partially or not at all. In order to compare MosaicSets to similar set visualization techniques, we also include Euler diagrams and frequency grids in Table~\ref{tab:taxonomy}.

\section{Expert Interviews}\label{apx:script}
\change{We also append the manuscript used for the expert interviews.   We performed one preliminary interview in March 2022 with two experts. We only consider their opinions on the rendering style from this interview. Later, in June 2022, we repeated an updated version of the interview with the two experts from the preliminary interview and an additional expert. The provided manuscript is the one we used for the second interview phase, but for completeness we also added the questions on the rendering from the first phase.}

\begin{figure*}
    \centering
    \begin{subfigure}{0.45\textwidth}
            \includegraphics[width=\textwidth]{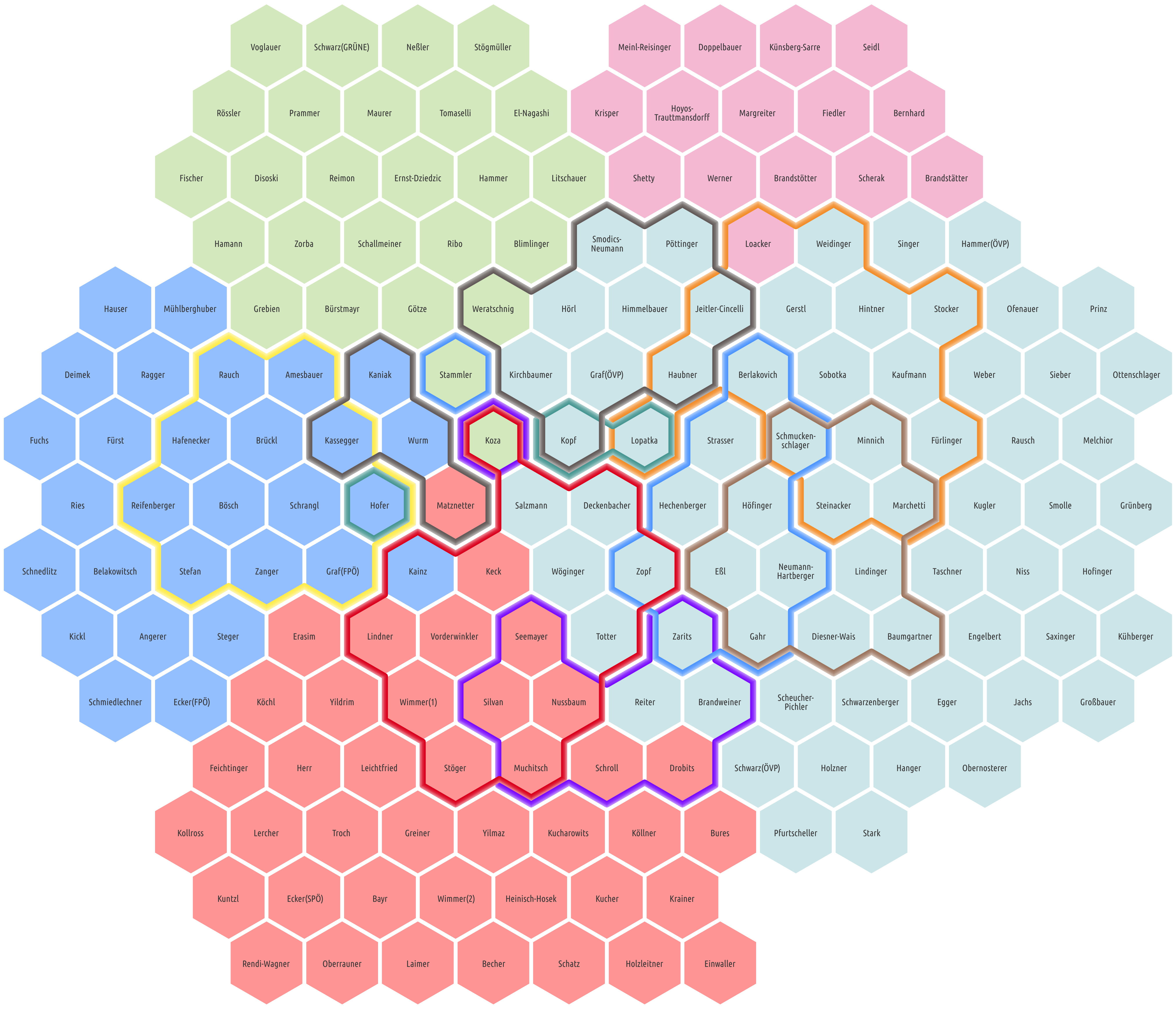}
            \caption{\datanationalrat{} with hexagonal grid}
    \end{subfigure}
    \hfill
    \begin{subfigure}{0.45\textwidth}
            \includegraphics[width=\textwidth]{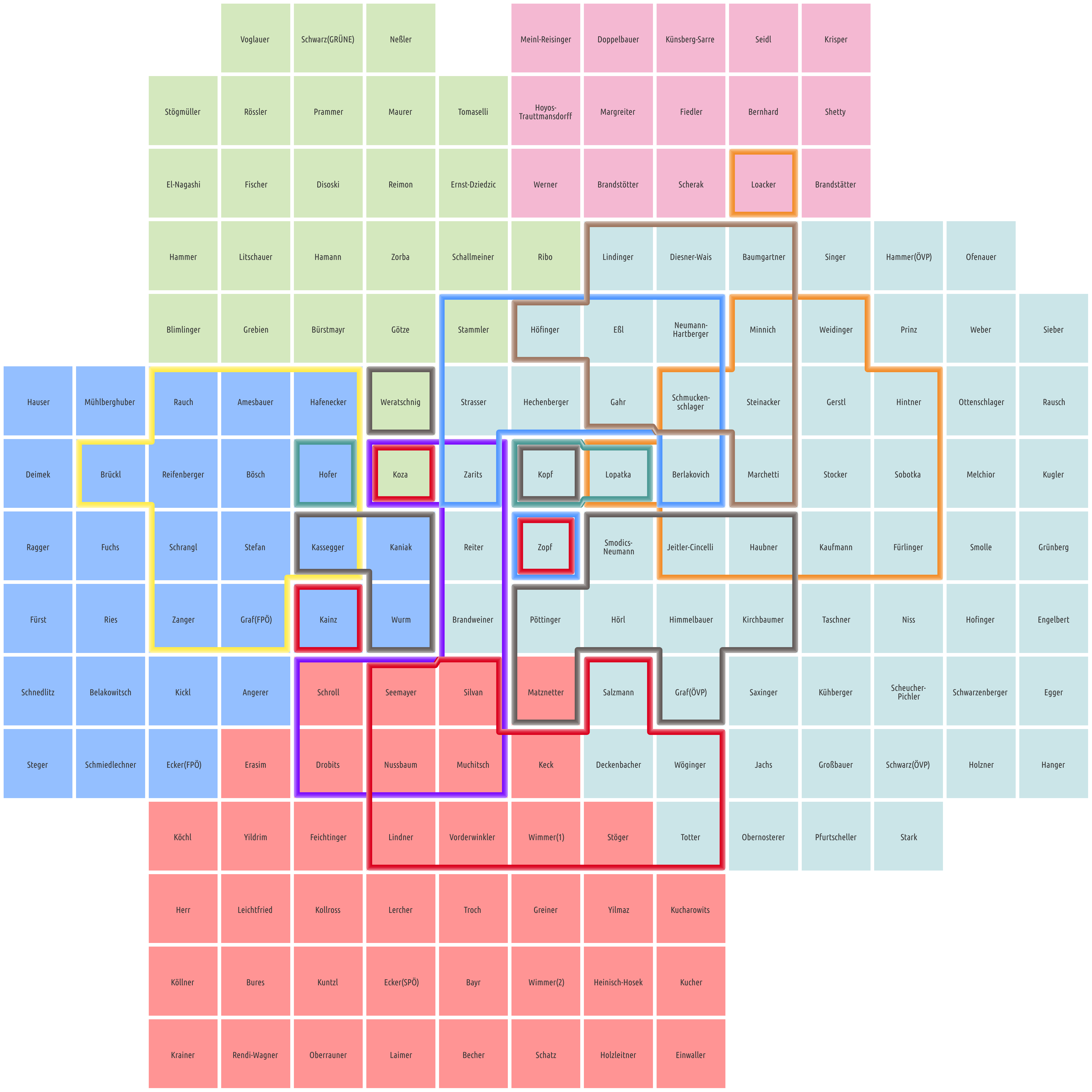}
            \caption{\datanationalrat{} with square grid}
    \end{subfigure}
    \caption{\change{\datanationalrat{} with eight overlay sets visualized using \MosaicSetsE{} and relaxing the contiguity constraint as described in Section 4.3. The solution shown in (a) has a $\pp{}=0.477$ and a computing time of $2.8$\,s. Three of the eight overlay sets are not contiguous (dark green, purple and blue). The solution shown in (b) has a $\pp{}=0.471$ and a computing time of $2.2$\,s. Five of the eight overlay sets are not contiguous (dark green, blue, red, orange and gray). \datanationalrat{} uses a color scheme that is a lighter variant of the typical political party colors used in Austria.}}
    \label{fig:relaxation}
\end{figure*}

\begin{figure*}[!htbp]
    \centering
    \begin{subfigure}{0.45\textwidth}
            \includegraphics[width=\linewidth]{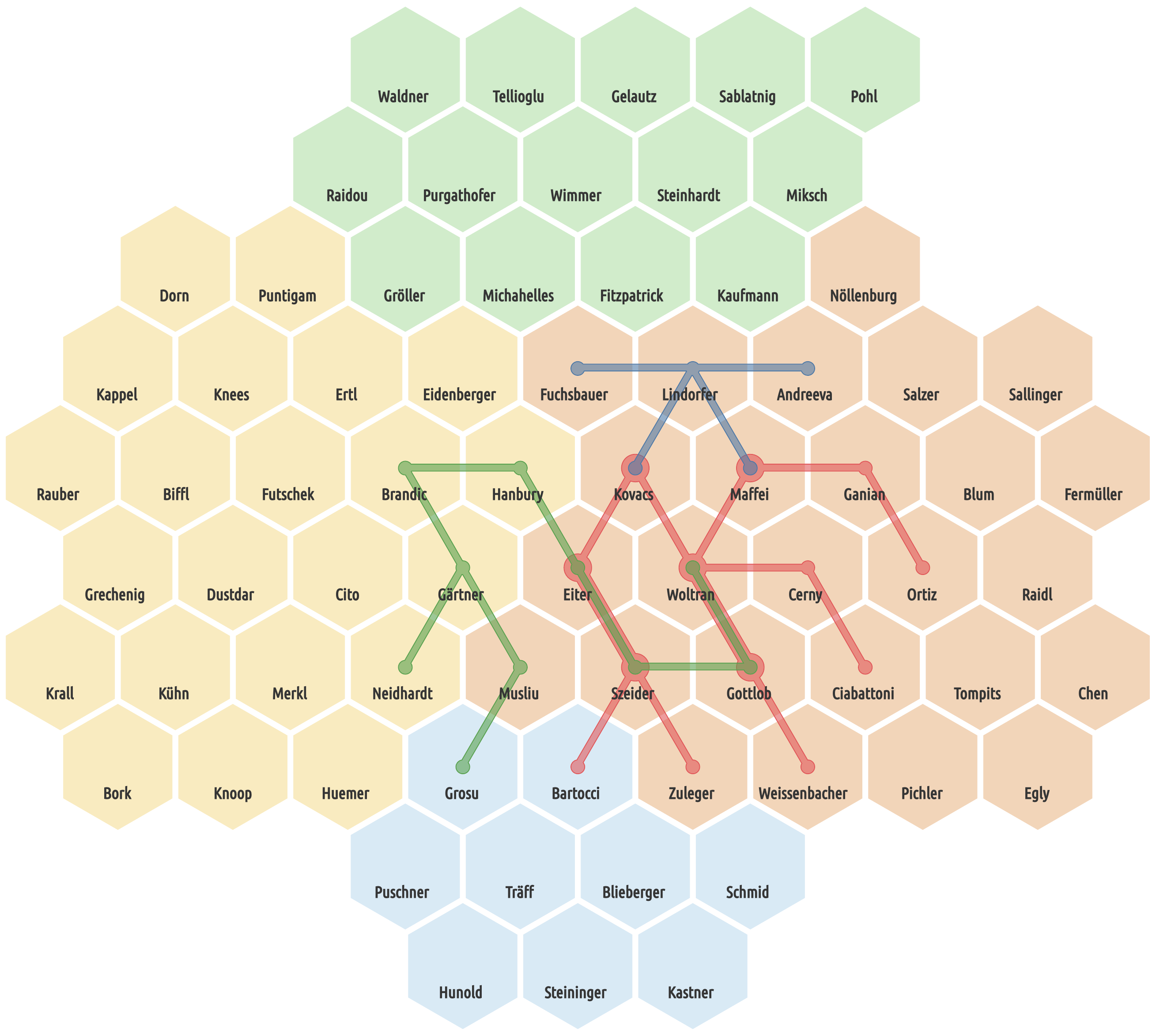}
            \caption{\datavienna{}}
    \end{subfigure}
    \hfill
    \begin{subfigure}{0.45\textwidth}
            \includegraphics[width=\linewidth]{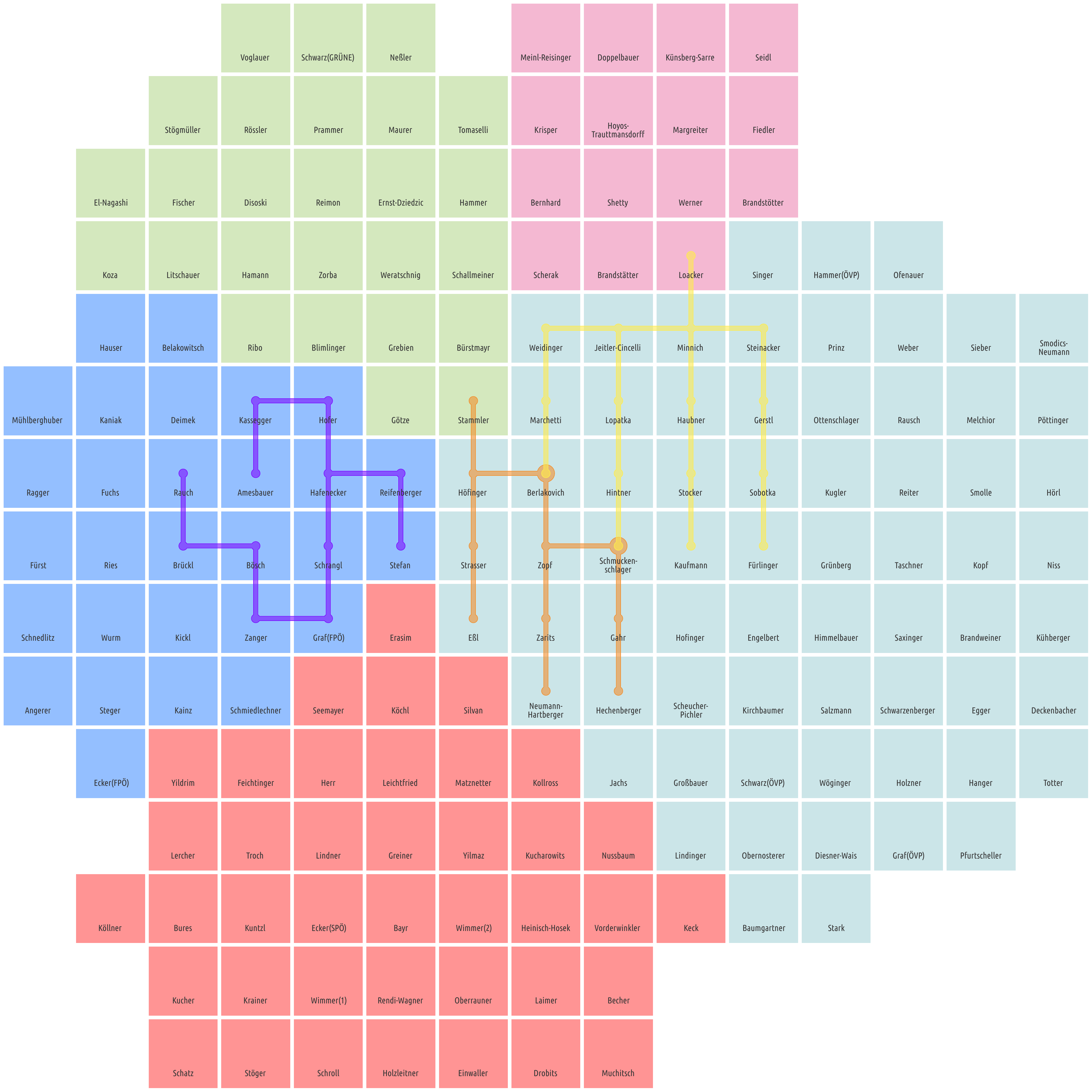}
            \caption{\datanationalrat{}}
    \end{subfigure}
    \caption{\datavienna{} (a) and \datanationalrat{} (b) with three overlay sets visualized with the Kelp style. \change{The corresponding visualizations with boundary style are given in Fig.~\ref{fig:V_sq}a and  Fig.~\ref{fig:parliament_sq}a, respectively. See also Fig.~\ref{fig:V_sq}a and  Fig.~\ref{fig:parliament_sq}a for compactness scores and running time. \datanationalrat{} uses a color scheme that is a lighter variant of the typical political party colors used in Austria.} }
    \label{fig:kelp_style} 
\end{figure*}

\begin{figure*}[!htbp]
    \centering
    \begin{subfigure}{0.45\textwidth}
            \includegraphics[width=\linewidth]{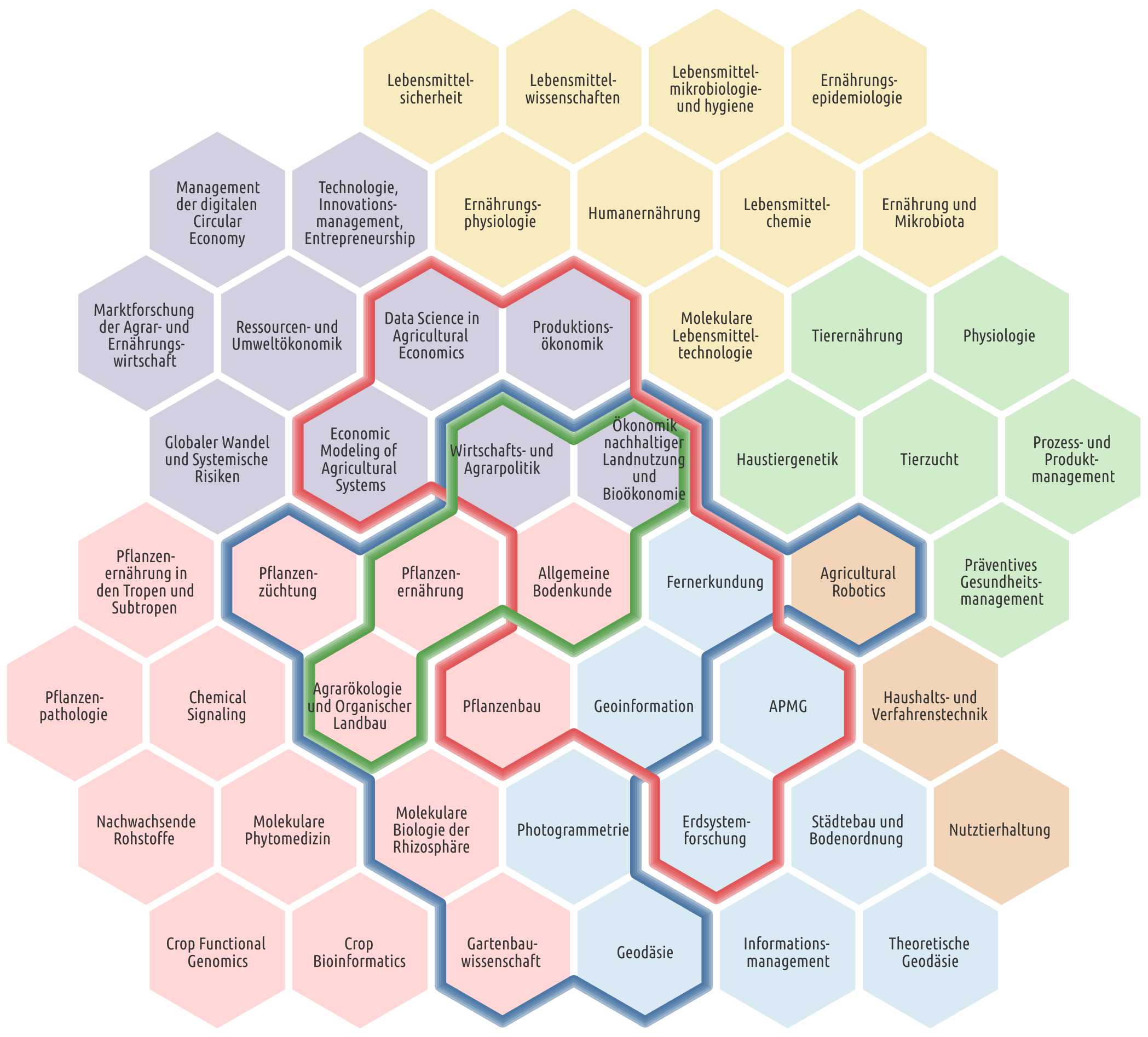}
            \caption{eccentricity-based compactness (\MosaicSetsE{})}
    \end{subfigure}
    \hfill
    \begin{subfigure}{0.45\textwidth}
            \includegraphics[width=\linewidth]{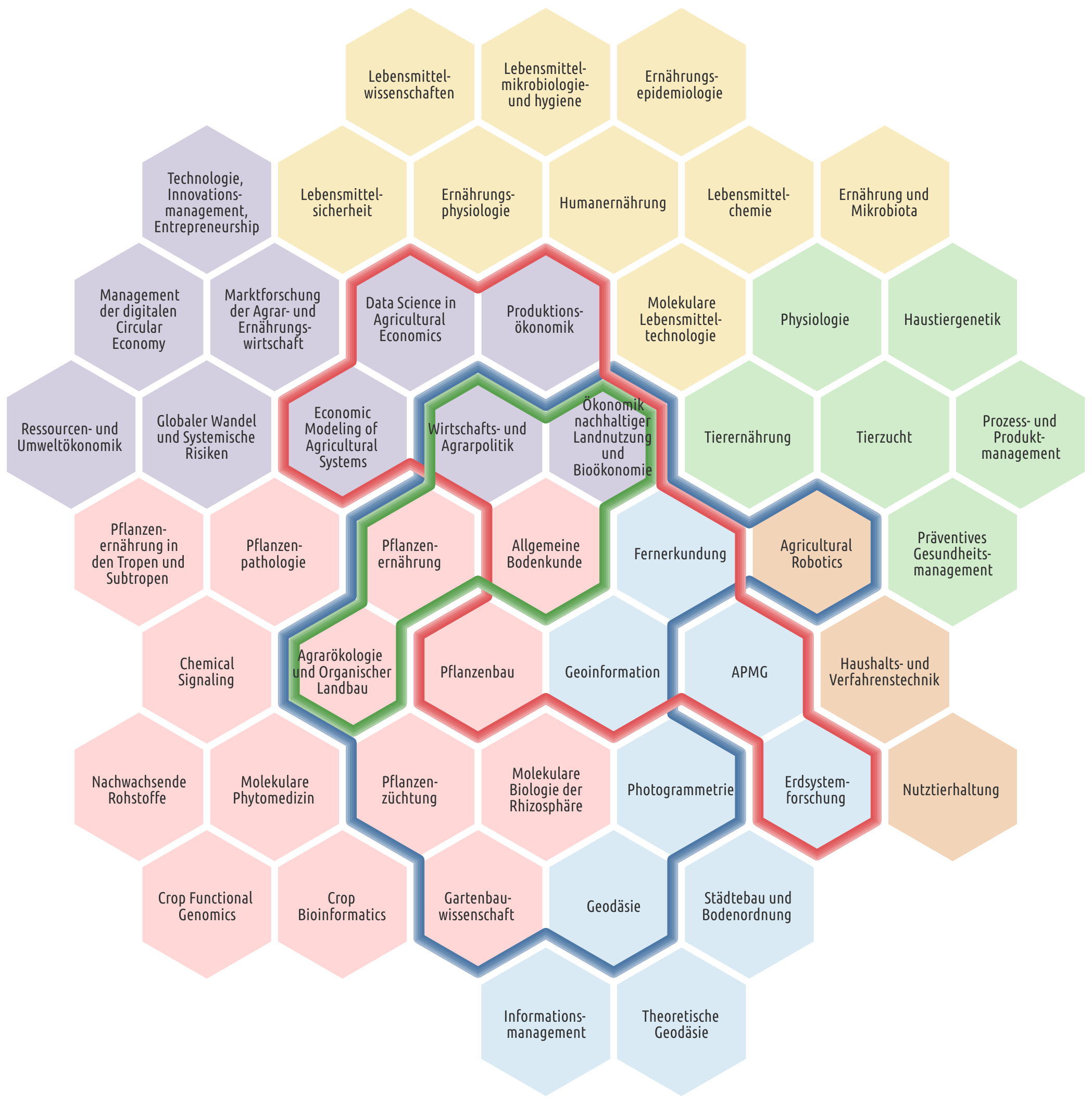}
            \caption{as (a) but area fixed after first iteration (\MosaicSetsEA{})}
    \end{subfigure}
    \caption{MosaicSets visualization for \databonn{} with three overlay sets and a hexagonal grid. (a) \MosaicSetsE{} with \change{$\pp{}=0.521$} and running time of \change{$0.7$\,s}. (b) \MosaicSetsEA{} with \change{$\pp{}=0.491$} and running time \change{$0.7$\,s}.}
    \label{fig:bn_hx} 
\end{figure*}

\begin{figure*}[!htbp]
    \centering
    \begin{subfigure}{0.45\textwidth}
            \includegraphics[width=\linewidth]{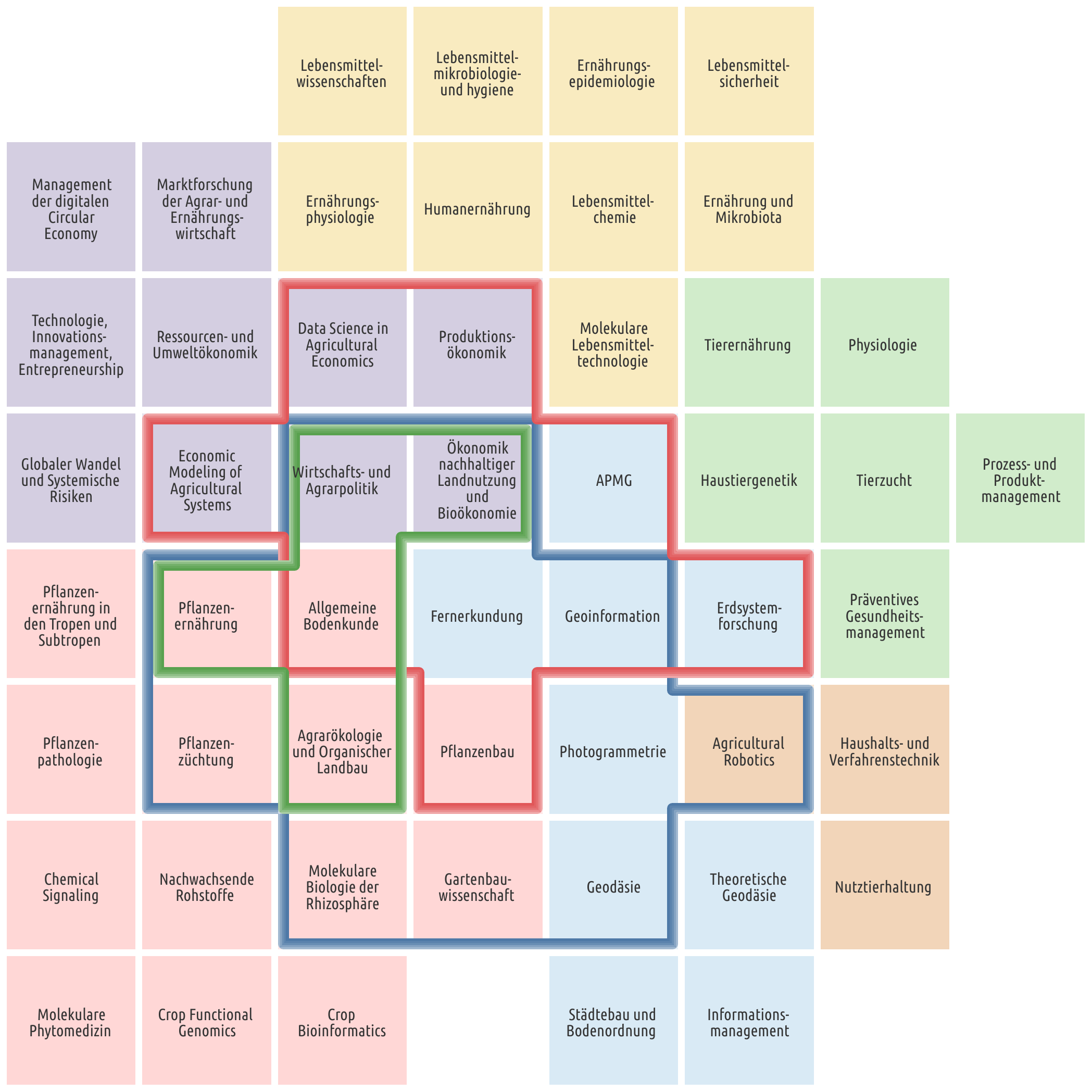}
            \caption{eccentricity-based compactness (\MosaicSetsE{})}
    \end{subfigure}
    \hfill
    \begin{subfigure}{0.45\textwidth}
            \includegraphics[width=\linewidth]{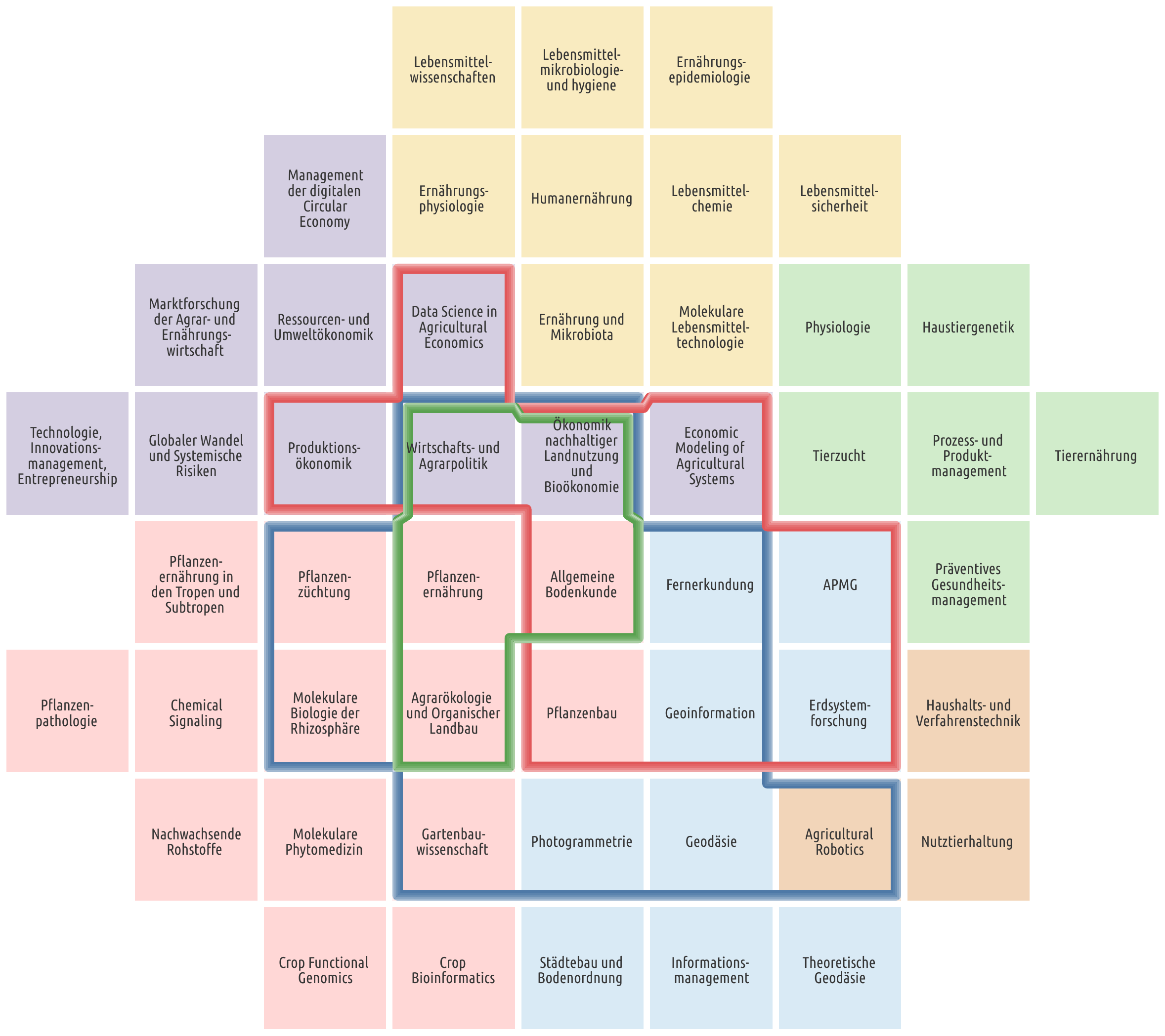}
            \caption{as (a) but area fixed after first iteration (\MosaicSetsEA{})}
    \end{subfigure}
    \caption{
    MosaicSets visualization for \databonn{} with three overlay sets and a square grid. (a) \MosaicSetsE{} with \change{$\pp{}=0.522$} and running time of \change{$0.6$\,s}. (b) \MosaicSetsEA{} with \change{$\pp{}=0.509$} and running time of \change{$0.8$\,s}.}
    \label{fig:bn_sq} 
\end{figure*}

\begin{figure*}[!htbp]
    \centering
    \begin{subfigure}{0.45\textwidth}
            \includegraphics[width=\linewidth]{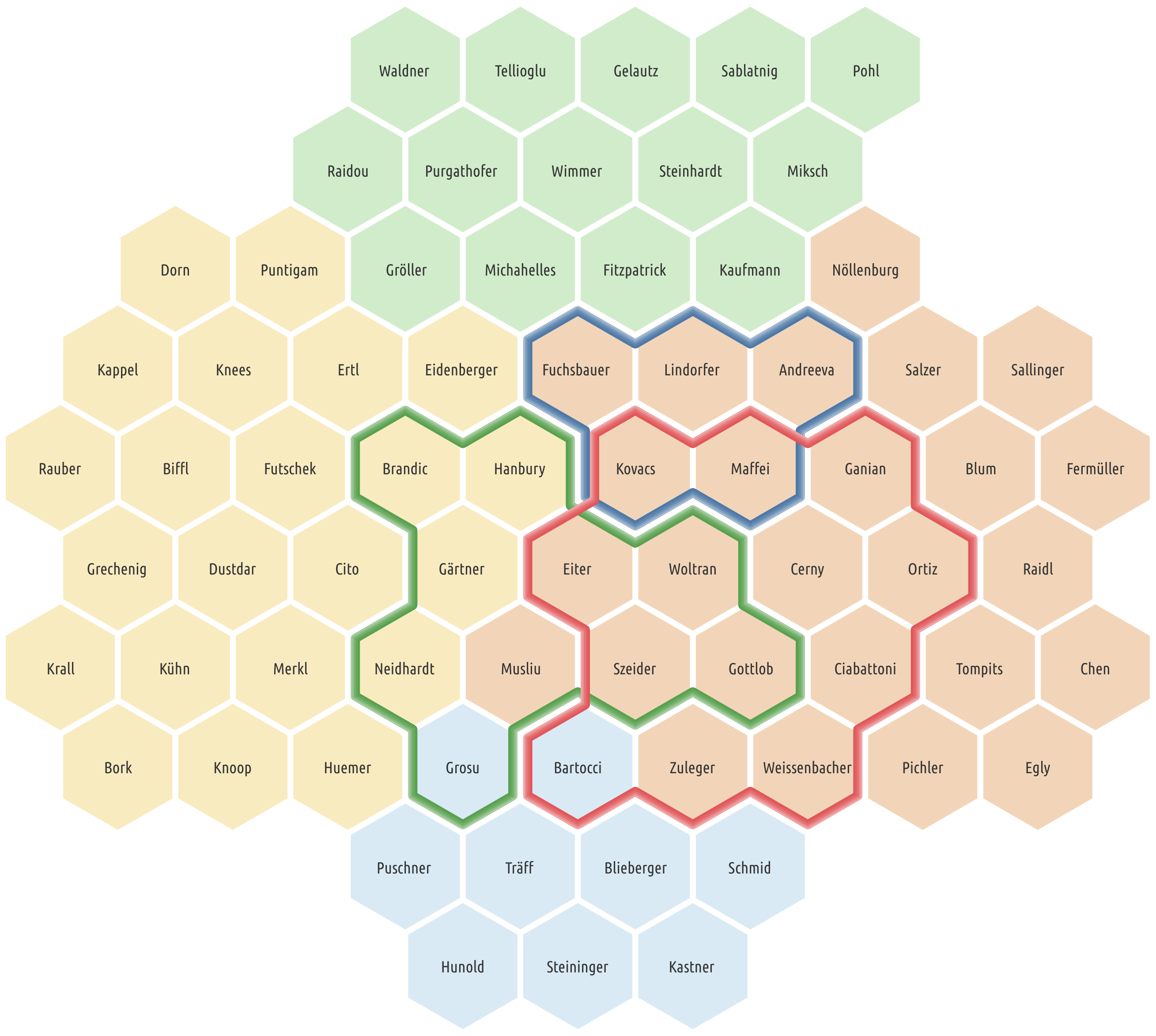}
            \caption{eccentricity-based compactness (\MosaicSetsE{})}
    \end{subfigure}
    \hfill
    \begin{subfigure}{0.45\textwidth}
            \includegraphics[width=\linewidth]{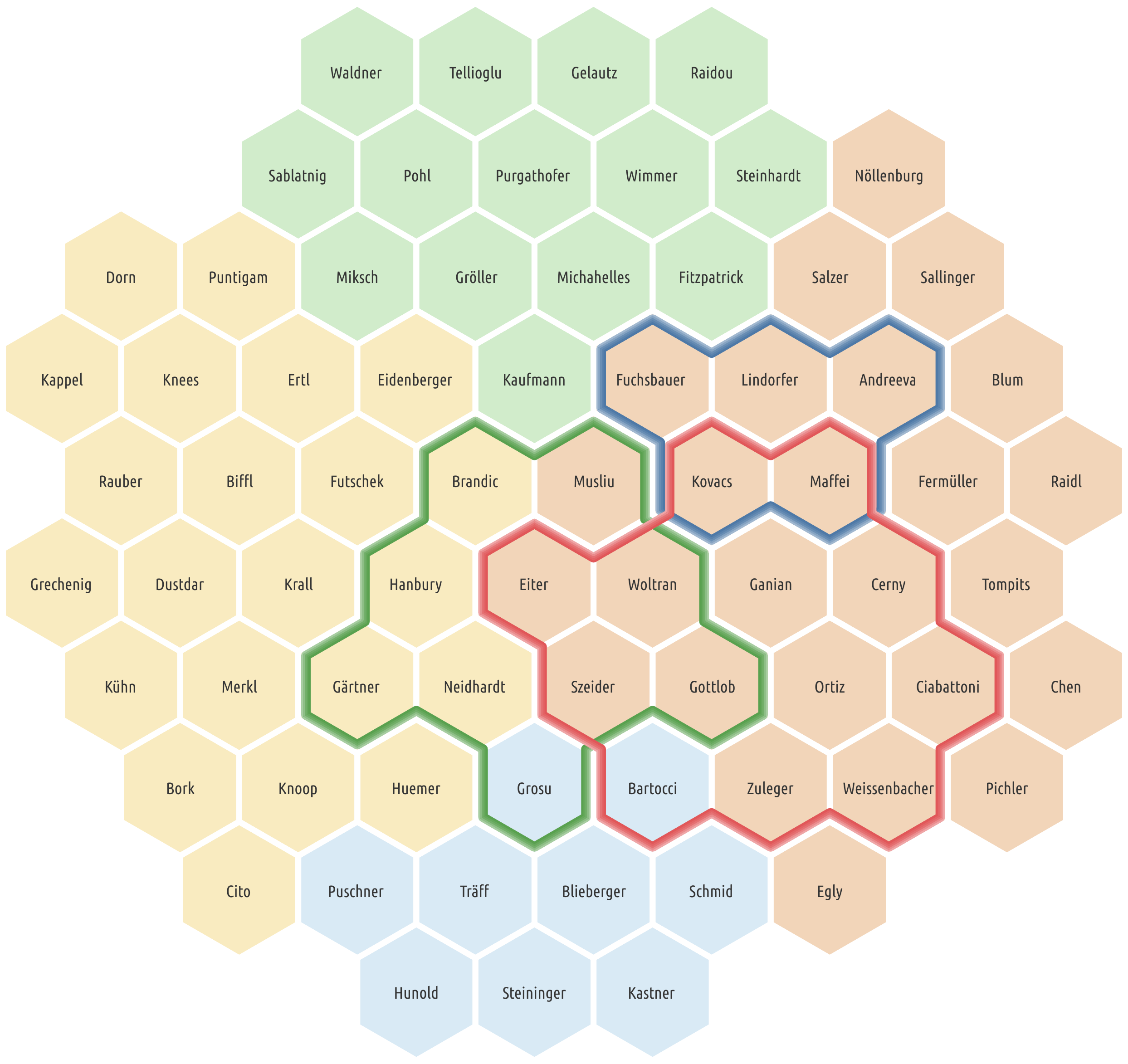}
            \caption{as (a) but area fixed after first iteration (\MosaicSetsEA{})}
    \end{subfigure}
    \caption{MosaicSets visualization for \datavienna{} with three overlay sets and a square grid. (a) \MosaicSetsE{} with \change{$\pp{}=0.574$} and running time of \change{$0.7$\,s}. (b) \MosaicSetsEA{} with \change{$\pp{}=0.576$} and running time of  \change{$0.7$\,s}.}
    \label{fig:V_hx} 
\end{figure*}

\begin{figure*}[!htbp]
    \centering
    \begin{subfigure}{0.45\textwidth}
            \includegraphics[width=\linewidth]{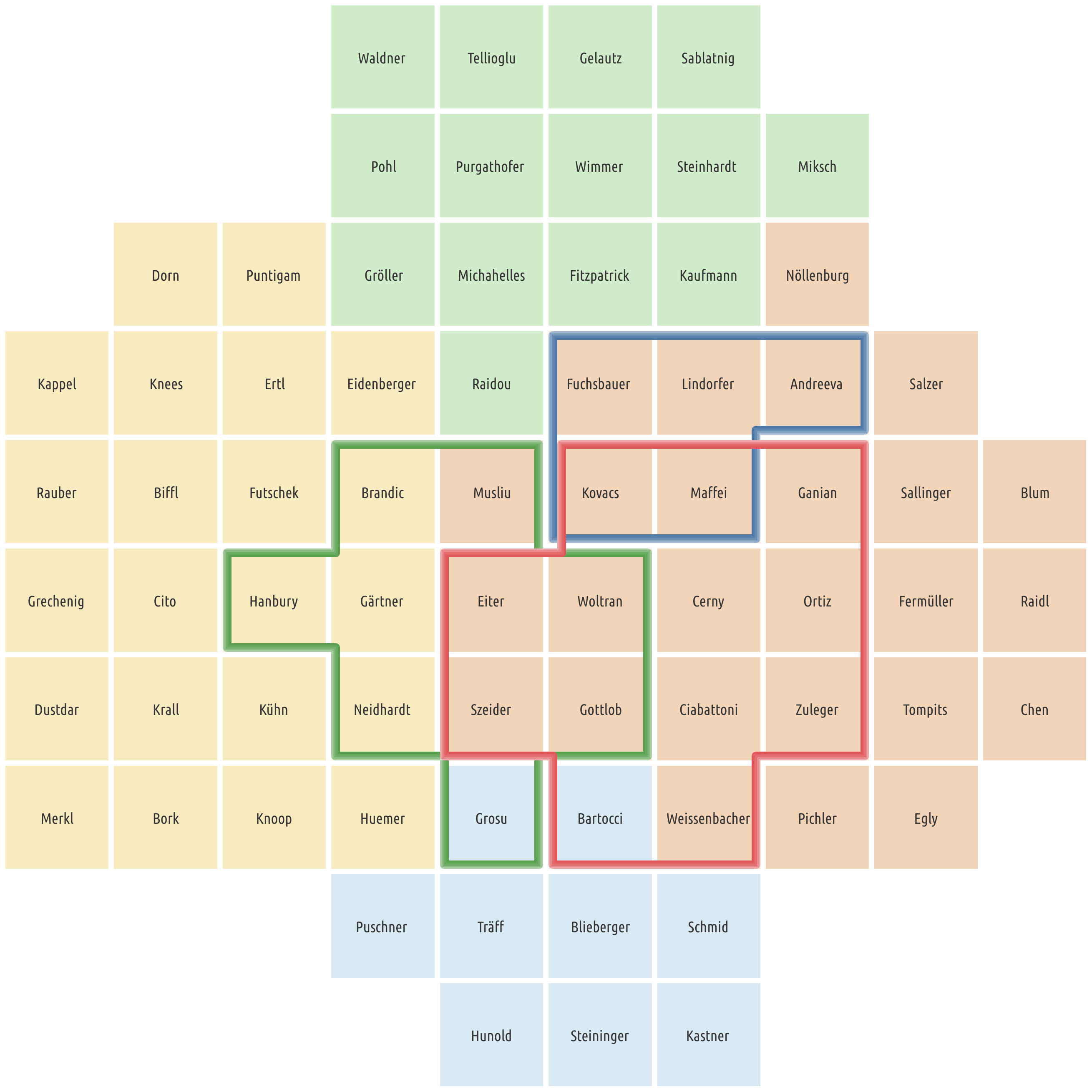}
            \caption{eccentricity-based compactness (\MosaicSetsE{})}
    \end{subfigure}
    \hfill
    \begin{subfigure}{0.45\textwidth}
            \includegraphics[width=\linewidth]{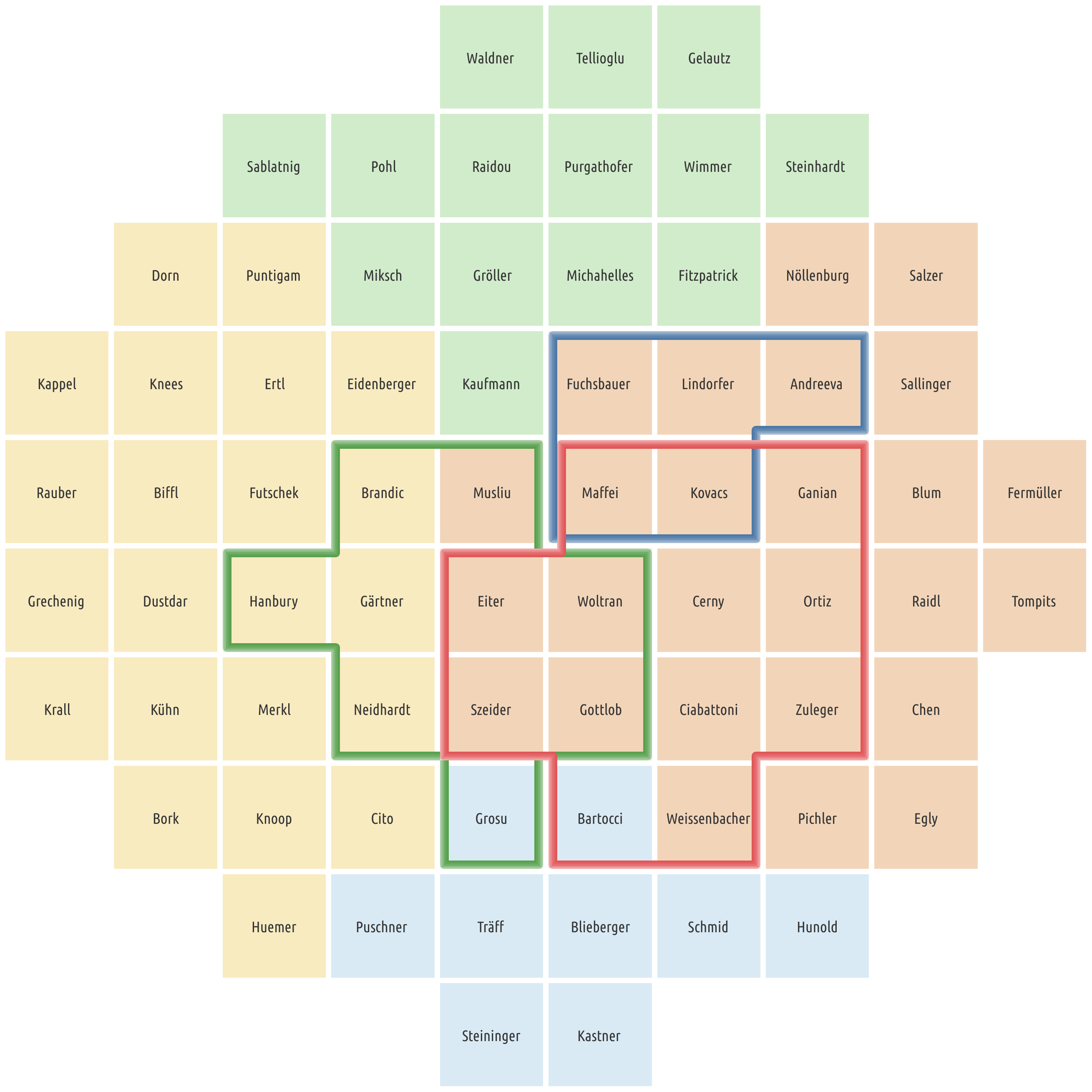}
            \caption{as (a) but area fixed after first iteration (\MosaicSetsEA{})}
    \end{subfigure}
    \caption{\datavienna{} with three overlay sets and square grid visualized with MosaicSets. (a) is \MosaicSetsE{} with \change{$\pp{}=0.584$} and running time of \change{$0.9$\,s}. (b)  \MosaicSetsEA{} with \change{$\pp{}=0.532$} and a running time of \change{$0.8$\,s}.}
    \label{fig:V_sq} 
\end{figure*}

\begin{figure*}[!htbp]
    \centering
    \begin{subfigure}{0.45\textwidth}
            \includegraphics[width=\linewidth]{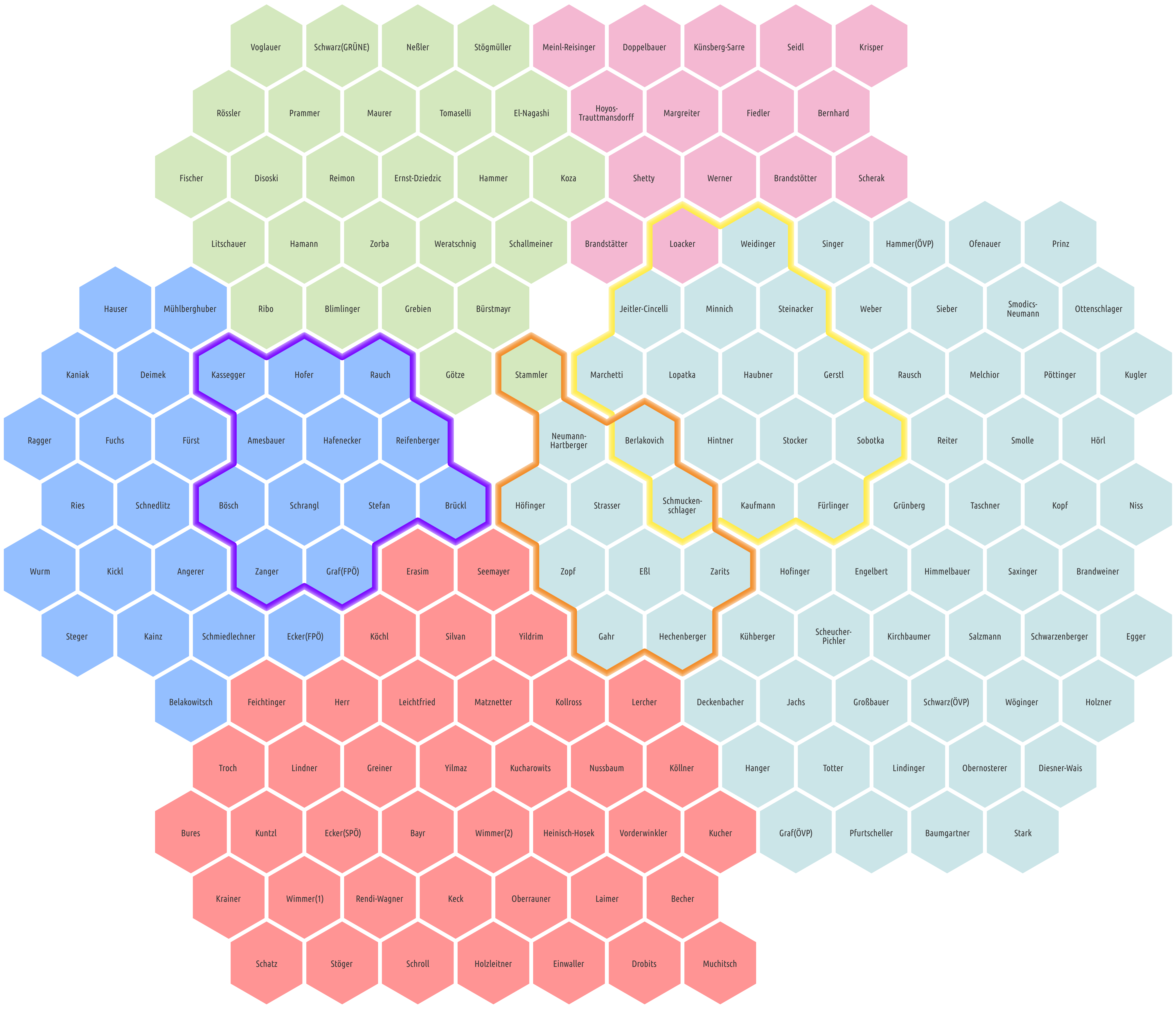}
            \caption{eccentricity-based compactness (\MosaicSetsE{})}
    \end{subfigure}
    \hfill
    \begin{subfigure}{0.45\textwidth}
            \includegraphics[width=\linewidth]{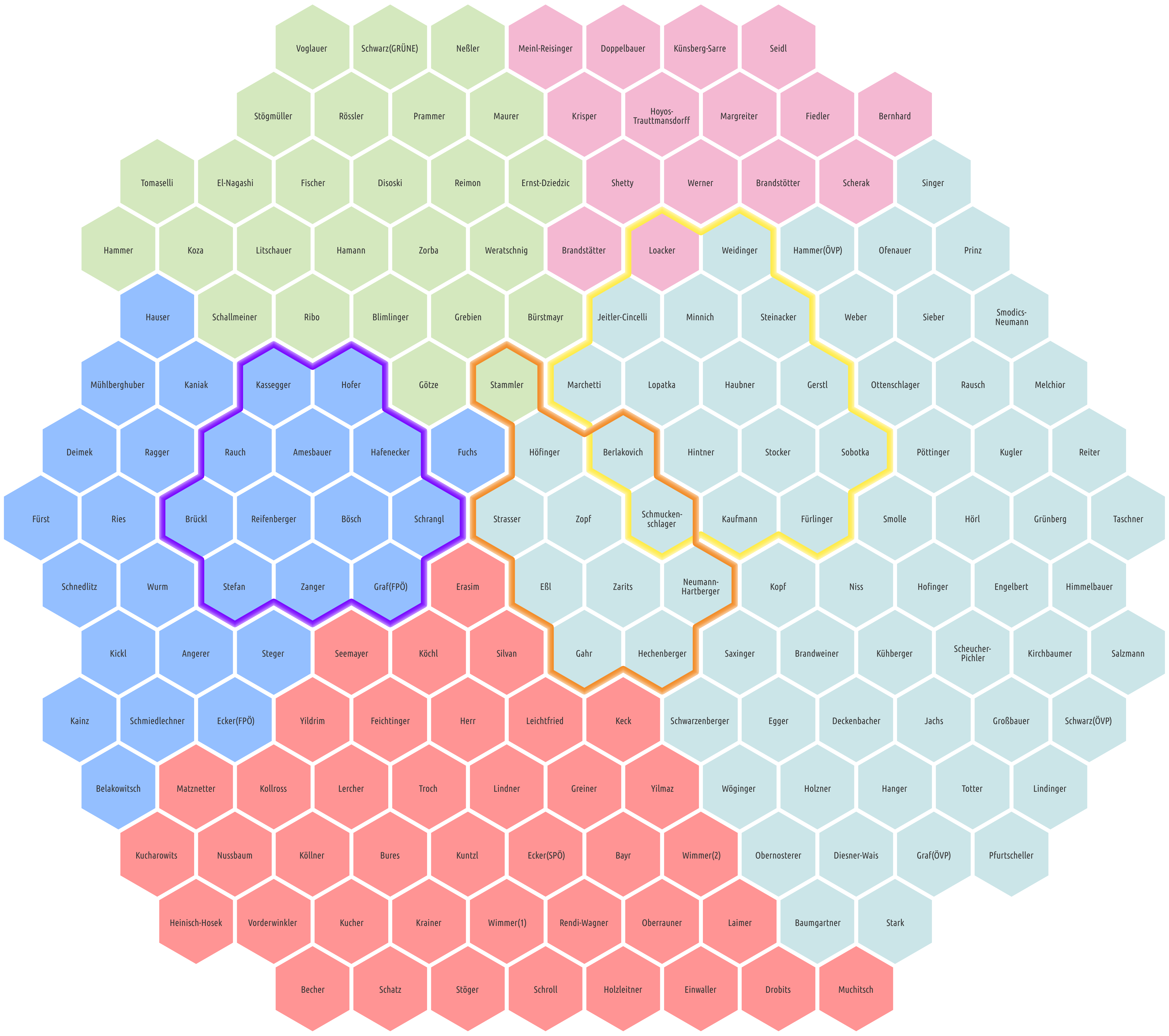}
            \caption{as (a) but area fixed after first iteration (\MosaicSetsEA{})}
    \end{subfigure}
    \caption{\datanationalrat{} with three overlay sets and hexagonal grid visualized with MosaicSets. (a) is \MosaicSetsE{} with \change{$\pp{}=0.563$} and running time of \change{$1.6$\,s}. (b)  \MosaicSetsEA{} with \change{$\pp{}=0.543$} and a running time of \change{$1.8$\,s. \datanationalrat{} uses a color scheme that is a lighter variant of the typical political party colors used in Austria.}}
    \label{fig:parliament_hx} 
\end{figure*}

\begin{figure*}[!htbp]
    \centering
    \begin{subfigure}{0.45\textwidth}
            \includegraphics[width=\linewidth]{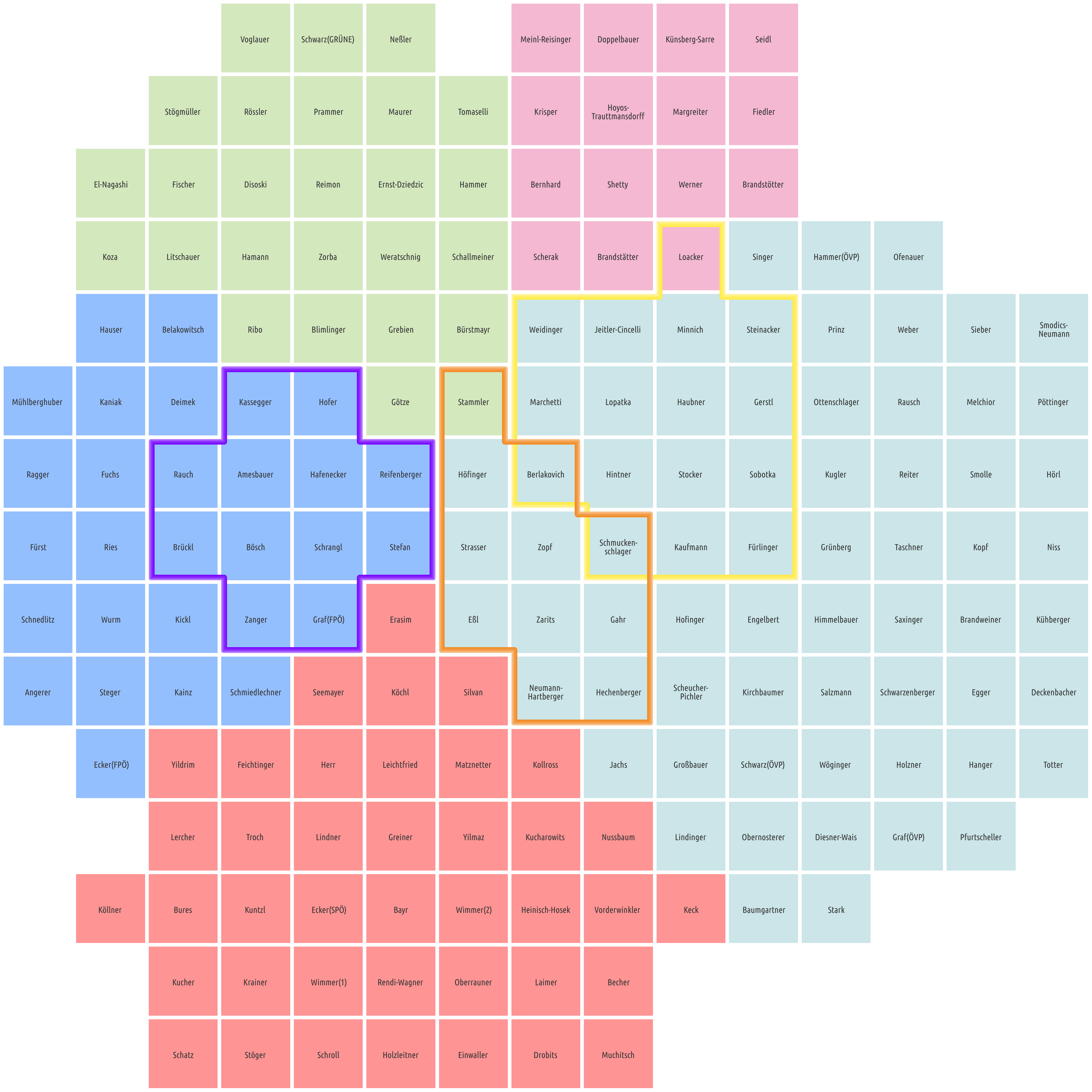}
            \caption{eccentricity-based compactness (\MosaicSetsE{})}
    \end{subfigure}
    \hfill
    \begin{subfigure}{0.45\textwidth}
            \includegraphics[width=\linewidth]{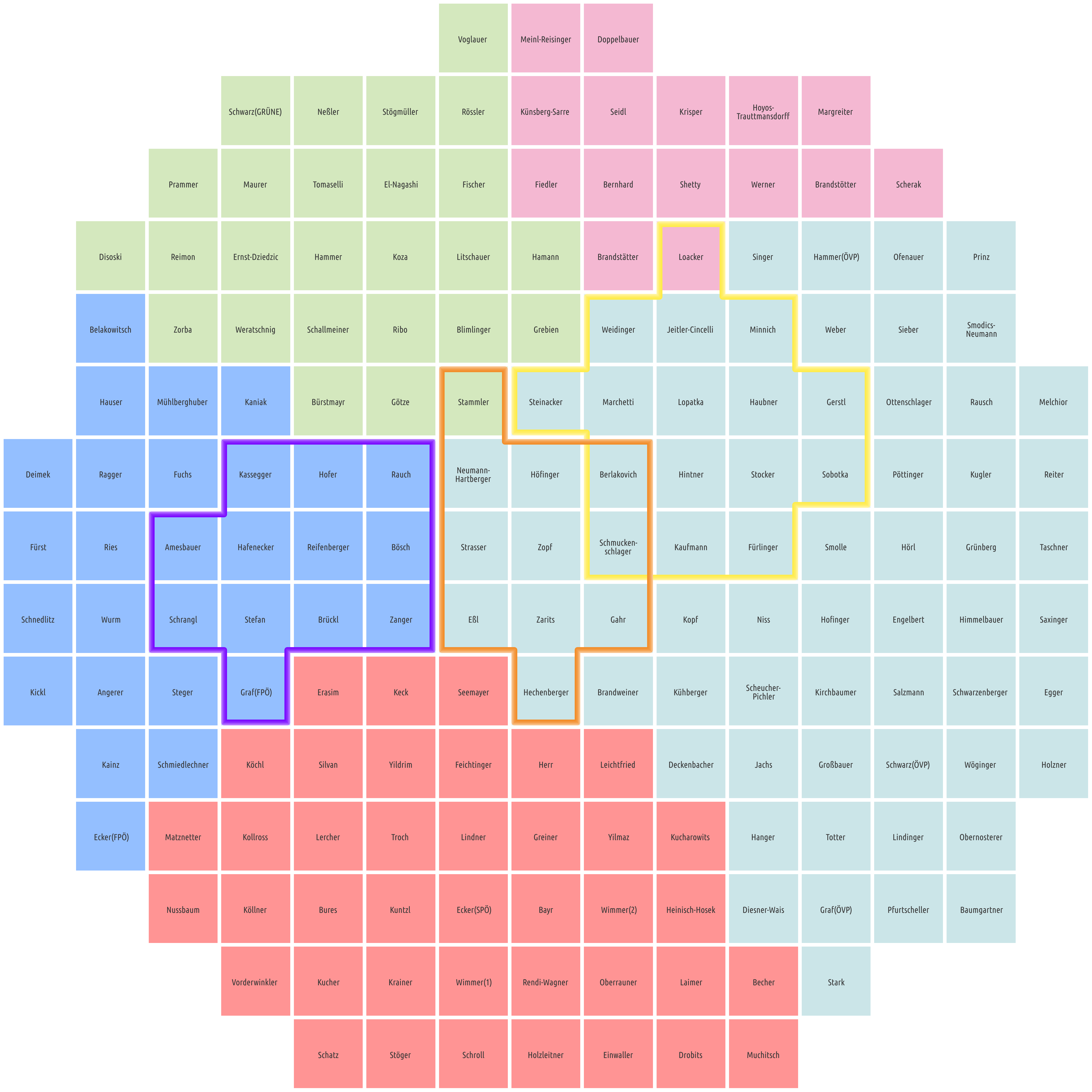}
            \caption{as (a) but area fixed after first iteration (\MosaicSetsEA{})}
    \end{subfigure}
    \caption{\datanationalrat{} with three overlay sets and square grid visualized with MosaicSets. (a) is \MosaicSetsE{} with \change{$\pp{}=0.579$} and running time of \change{$1.6$\,s}. (b)  \MosaicSetsEA{} with \change{$\pp{}=0.512$} and a running time of \change{$1.8$\,s. \datanationalrat{} uses a color scheme that is a lighter variant of the typical political party colors used in Austria.}}
    \label{fig:parliament_sq} 
\end{figure*}

\begin{table*}[!ht]
\resizebox{\textwidth}{!}{%
\begin{tabular}{c|l|c|c|c|c|c|}
\textbf{\#} & \multicolumn{1}{c|}{\textbf{Task}}& 
\begin{tabular}{@{}c@{}}MosaicSets \\ Expert 1\end{tabular} &
\begin{tabular}{@{}c@{}}MosaicSets \\ Expert 2\end{tabular} &
\begin{tabular}{@{}c@{}}MosaicSets \\ Expert 3\end{tabular} & \begin{tabular}{@{}c@{}}Euler \\ Diagrams\end{tabular} & \begin{tabular}{@{}c@{}}Frequency \\ Grids\end{tabular} \\ \hline
A1        & Find/Select elements that belong to a specific set                                 & $\bullet$ & $\bullet$ & $\bullet$ & $\bullet$ & $\bullet$               \\
A2        & Find sets containing a specific element.                                           & $\bullet$ & $\bullet$ & $\circ$ & $\bullet$ & $\bullet$                \\
A3        & Find/Select elements based on their set memberships                                & $\bullet$ & $\bullet$ & $\bullet$ & $\bullet$ & $\circ$                     \\
A4        & Find/Select elements in a set with a specific set member-ship degree               & $\circ$ & $\bullet$ & $\bullet$ & - & $\circ$                   \\
A5        & Filter out elements based on their set memberships.                                & - & - & - & - & $\circ$              \\
A6        & Filter  out  elements  based  on  their  set  membership  degrees                  & - & - & - & - & $\circ$                 \\
A7        & Create a new set that contains certain elements.                                   & - & - & - & $\circ$ & $\circ$                 \\
\hline
B1        & Find out the number of sets in the set family.                                      & $\bullet$ & $\bullet$ & $\bullet$ & $\circ$ & $\circ$             \\
B2        & Analyze inclusion relations.                                                      & $\bullet$  & $\bullet$ & $\bullet$ & $\bullet$ & $\bullet$              \\
B3        & Analyze  inclusion  hierarchies                                                    & $\bullet$ & $\bullet$ & $\bullet$ & $\bullet$ & -                \\
B4        & Analyze exclusion relation                                                         & $\bullet$ & $\bullet$ & $\bullet$ & $\bullet$ & -          \\
B5        & Analyze intersection relation                                                      & $\bullet$ & $\bullet$ & $\bullet$ & $\bullet$ & $\bullet$              \\
B6        & Identify intersections between k sets                                              & $\bullet$ & $\bullet$ & $\bullet$ & $\circ$ & $\bullet$          \\
B7        & Identify the sets involved in a certain intersection                               & $\bullet$ & $\bullet$ & $\bullet$& $\bullet$ & $\bullet$           \\
B8        & Identify set intersections belonging to a specific set                             & $\bullet$ & $\bullet$ & $\bullet$ & $\bullet$ & $\bullet$            \\
B9        & Identify the set with the largest / smallest number of pair-wise set intersections & $\bullet$ & $\bullet$ & $\bullet$ & $\circ$ & $\circ$             \\
B10       & Analyze and compare set- and intersection cardinalities                            & $\bullet$ & $\bullet$ & $\bullet$ & $\bullet$ & $\bullet$         \\
B11       & Analyze and compare set similarities                                               & -  & - & - & $\circ$ & -            \\
B12       & Analyze and compare set exclusiveness                                              & $\bullet$ & $\bullet$ & $\bullet$ & $\circ$ & $\bullet$            \\
B13       & Highlight  specific  sets,  subsets,  or  set  relations                           & -  & - & - & n/a & $\bullet$         \\
B14       & Create a new set using set-theoretic operation                                     & -  & - & - & $\circ$ & -             \\
\hline
C1        & Find out the attribute values of a certain element                                 & - & - & - & $\circ$ & \small{\faMousePointer}             \\
C2        & Find out the distribution of an attribute in a certain set or subset               & - & -  & - & $\circ$ & -            \\
C3        & Compare the attribute values between two sets or subsets                           & - & -  & - & $\circ$ & -              \\
C4        & Analyze  the  set  memberships  for  elements  having  certain attribute values    & - & - & - & - & -             \\
C5        & Create  a  new  set  out  of  elements  that  have  certain  attribute values      & -& - & -  & -& -              \\ %
\end{tabular}%
}
\caption{Task taxonomy by Alsallakh et al.~\cite{DBLP:journals/cgf/AlsallakhMAHMR16} and assessment of MosaicSets, Euler diagrams and frequency grids~\cite{micallef2012assessing}. Tasks can be classified as: ($\bullet$) supported; ($\circ$) partially supported; (-) unsupported; or  ({\small\faMousePointer}) requires interactivity. For MosaicSets the classification is based on the opinion of the expert interviews. For Euler diagrams and frequency grids the classification is according to~\cite{DBLP:journals/cgf/AlsallakhMAHMR16}.}
\label{tab:taxonomy}
\end{table*}

\bibliographystyle{abbrv-doi-hyperref}

\bibliography{references}

\AtEndDocument{
\includepdf[pages=-, nup= 2x2, frame=true, delta=3mm 3mm, scale=0.9] {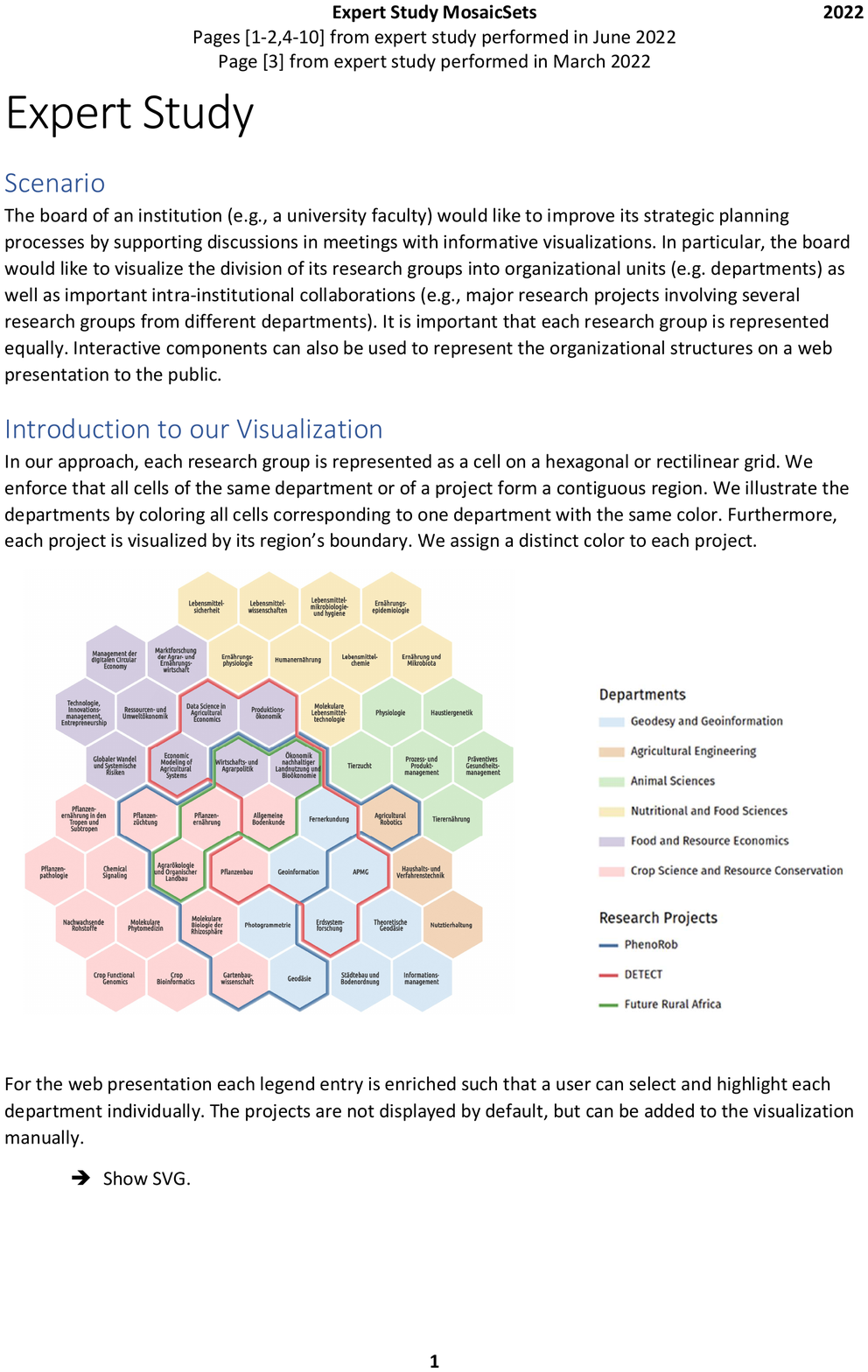}}